\documentclass[11pt]{article}
\usepackage[T1]{fontenc}
\usepackage{lmodern}
\usepackage{authblk}
\usepackage[margin=1.1in]{geometry}
\usepackage{amsmath, amssymb, amsthm, mathtools}
\usepackage{enumitem}
\usepackage{graphicx}
\usepackage{microtype}

\makeatletter
\renewcommand{\paragraph}{%
  \@startsection{paragraph}{4}{\z@}%
    {-3.25ex \@plus -1ex \@minus -.2ex}%
    {1.5ex \@plus .2ex}%
    {\normalfont\normalsize\bfseries}}
\makeatother

\newtheorem{theorem}{Theorem}[section]
\newtheorem{lemma}[theorem]{Lemma}
\newtheorem{proposition}[theorem]{Proposition}
\newtheorem{corollary}[theorem]{Corollary}
\theoremstyle{remark}
\newtheorem{remark}[theorem]{Remark}

\usepackage[authoryear]{natbib}
\setcitestyle{authoryear,open={(},close={)}}
\usepackage{pgfplots}
\pgfplotsset{compat=1.17}
\usepgfplotslibrary{groupplots}
\usepackage[hidelinks, hypertexnames=false, pdftitle={Bi-Compositional Division Rules}, pdfauthor={Christoph Schlegel}]{hyperref}

\definecolor{s1}{named}{teal}
\colorlet{s2}{blue!70!black}
\colorlet{s3}{violet!80!black}
\colorlet{s4}{red!75!black}

\pgfplotsset{
  gainpanel/.style={
    width=3.55cm, height=3.7cm,
    xmin=0, xmax=2.12, ymin=0, ymax=2.12,
    xtick=\empty, ytick=\empty,
    axis lines=left, tick label style={font=\scriptsize},
    title style={font=\small, yshift=-2pt},
    every axis plot/.append style={semithick},
    samples=100, domain=0:2,
  }
}

\newcommand{\R}{\mathbb{R}}
\newcommand{\CEA}{\mathrm{CEA}}
\newcommand{\CEL}{\mathrm{CEL}}
\newcommand{\PROP}{\mathrm{PROP}}
\newcommand{\id}{\mathrm{id}}
\newcommand{\norm}[1]{\lVert #1 \rVert}
\newcommand{\T}{\mathcal{T}}
\newcommand{\Lf}{\mathcal{L}}

\title{Bi-Compositional Division Rules\thanks{The manuscript was co-written with Anthropic's Claude Opus and Fable models and OpenAI's GPT Astra model. The proof ideas were AI-generated. Human input was used for posing the research question, providing context, verification of the results, simplification of the proofs and for the write-up, redaction and organisation of the paper. Extensive comments by William Thomson on several drafts, as well as comments by Leo Arias and Bruno Mazorra on an earlier draft are gratefully acknowledged.}}
\author{Christoph Schlegel}
\affil{Flashbots, Zurich, Switzerland}
\date{\today}

\begin{document}
\maketitle

\begin{abstract}
We characterise the division rules for claims problems that satisfy bilateral consistency, composition down and composition up. The claimants form ordered priority classes. Each class is paid in full before any lower class receives an award. A class of at least three claimants uses a log-exponential rule with claimant-specific parameters, weighted constrained equal awards (CEA), or weighted constrained equal losses (CEL). A two-claimant class admits any bilateral difference rule in which neither claimant has absolute priority. 

For a parameter profile $\theta=(\theta_i)_{i\in N}$, the log-exponential rule equalises sacrifice in awards for the utilities $u_{\theta_i}$ and in losses for the dual utilities $u_{-\theta_i}$, where $u_\theta(x)=\log\varphi_\theta(x)$ for $x>0$ and
\[
\varphi_\theta(x):=\frac{e^{\theta x}-1}{\theta}\quad(\theta\ne0),
\qquad \varphi_0(x):=x.
\]

Adding equal treatment of equals forces a single priority class. Moreover, the log-exponential rule then uses one parameter $\theta$ for all claimants and we obtain a symmetric family ranging from CEA at $\theta=-\infty$ to CEL at $\theta=+\infty$, with proportional division at $\theta=0$.
\end{abstract}

\section{Introduction}
Dividing when there is not enough or when a burden must be shared is a ubiquitous problem: the liquidation value of an insolvent firm is divided among its creditors, a revenue target among taxpayers, and losses from a default among the counterparties of a clearing house.

Many familiar rules, including proportional division, constrained equal losses and constrained equal awards, are homogeneous by construction: if all claims and the endowment are multiplied by a common factor, all awards are multiplied by the same factor. Homogeneity requires that the division depend only on the relative claims and the relative endowment, so that a rule keying on absolute size must invoke a reference magnitude that the problem itself does not supply. Equivalently, the same division principle applies at every absolute scale. The problems with claims and endowment $(50,30;40)$ and $(5,3;4)$ are different, and homogeneity ties the division of the larger to the division of the smaller. It therefore rules out fixed absolute thresholds and other forms of size dependence, which may be part of what a rule is meant to capture.\footnote{Homogeneity is sometimes read instead as a statement about measurement, namely that recommendations on how to divide should not depend on whether amounts are recorded in dollars or in cents. That reading is largely incorrect: a change of the unit of account rescales every quantity of the same dimension, including any parameter of the rule, and the rescaled data describe the same problem.}%

In taxation the case for relaxing homogeneity is especially clear, and most schedules used in practice are not homogeneous. Two path-independence requirements are natural in this setting. Composition down requires that if the endowment turns out to be smaller than announced, dividing the revised endowment directly gives the same awards as dividing it again among the claimants, the amounts already assigned to them now serving as their claims. Its counterpart, composition up, requires that if the endowment turns out to be larger, dividing the increment among the unmet parts of the claims gives the same awards as dividing the larger endowment directly. In this context, \cite{Young1988} introduced the equal-sacrifice rules: each such rule is specified by a continuous strictly increasing function $u$ and equalises the sacrifices $u(c_i)-u(x_i)$ across claimants. The function $u$ rescales the axis on which claims and awards are measured. \cite{Young1988} characterised the equal-sacrifice rules by \emph{equal treatment of equals, continuity, bilateral consistency, and composition down} together with two strictness conditions.\footnote{\citet{ChambersMorenoTernero2017} remove the strictness conditions and characterise the resulting class of generalised equal-sacrifice rules; see Section~\ref{sec:literature}.}

The author of the present paper encountered the problem of designing non-homogeneous division rules in another context: auto-deleveraging on a futures exchange. A large price move can leave a trading account unable to cover its losses, so that the counterparties holding the profitable positions cannot be paid in full; the exchange reduces those positions at prices worse than the market price, and the question is how the reduction is distributed among them.\footnote{In the companion paper \citep{SchlegelMamageishvili2026}, an auto-deleveraging problem specifies each account's position size and equity on one side of a futures market, together with a target for total open interest on that side after deleveraging. The rule assigns residual positions between zero and the original position sizes that sum to the target. The position sizes play the role of claims, and equity gives the rule an additional way to distinguish between accounts with the same claim. We study the trade-off between reducing platform risk and distributing forced position reductions fairly within the class of bi-compositional rules characterised here. We show how to calibrate a rule to a risk bound using historical data and evaluate the rules by replaying Hyperliquid's auto-deleveraging episode of 10 October 2025.} It has been argued that a rule in this context should prioritise reducing risk across positions~\citep{ADL}, which leads to an equity-weighted version of the constrained equal awards rule. A rule that gives full weight to risk reduction can allocate losses quite unevenly. If risk is homogeneous and risk management is the primary concern, this cannot be avoided.
If risk is non-homogeneous, however, then deleveraging the riskiest positions first may be warranted if those positions are highly leveraged, but less so when they are only moderately leveraged.

In both applications, it is useful to have a family of rules that can be adjusted continuously to reflect the desired degree of progressivity and risk sensitivity, respectively. We introduce such a family and give it an axiomatic foundation. In this family, which we call log-exponential, each claimant $i$ has a parameter $\theta_i\in\R$ and a utility $u_{\theta_i}$ in which to measure sacrifice. At an interior endowment the rule equalises $u_{\theta_i}(c_i)-u_{\theta_i}(x_i)$ across claimants with positive claims. Here $u_\theta(x)=\log\varphi_\theta(x)$ for $x>0$, with
\[
\varphi_\theta(x):=\frac{e^{\theta x}-1}{\theta}\quad(\theta\ne0),
\qquad
\varphi_0(x):=x.
\]
Equivalently, $\varphi_{\theta_i}(x_i)=\lambda\varphi_{\theta_i}(c_i)$, where the common budget parameter $\lambda\in[0,1]$ is chosen so that awards sum to the endowment. By duality, the rule simultaneously equalises sacrifice in losses for the utilities $u_{-\theta_i}$.

Using the same parameter for every claimant gives the symmetric family $\{r^\theta\}_{\theta\in[-\infty,+\infty]}$. It interpolates smoothly between CEA at $\theta=-\infty$ and CEL at $\theta=+\infty$, with proportional division at $\theta=0$.\footnote{The endpoints CEA and CEL are limiting members. The equal-sacrifice utility formula applies only for finite $\theta$.}
Figure~\ref{rtheta} shows the award paths for a two-claimant example.
\begin{figure}[t]
\centering
\begin{tikzpicture}
\begin{axis}[width=10cm, height=4.5cm, xmin=0, xmax=2.72, ymin=0, ymax=1.13,
  xtick={0,0.5,1,1.5,2,2.5}, ytick={0,0.5,1},
  axis lines=left, tick label style={font=\scriptsize},
  xlabel={award to claimant 1}, ylabel={award to claimant 2},
  label style={font=\small}, every axis plot/.append style={semithick},
  clip=false]
  \addplot[densely dashed, s1] coordinates {(0.0000,0.0000) (1.0000,1.0000) (2.5000,1.0000)};
  \addplot[s1] coordinates {(0.0000,0.0000) (0.0125,0.0124) (0.0250,0.0248) (0.0375,0.0372) (0.0500,0.0496) (0.0625,0.0620) (0.0750,0.0744) (0.0875,0.0868) (0.1000,0.0991) (0.1125,0.1115) (0.1250,0.1238) (0.1375,0.1362) (0.1500,0.1485) (0.1625,0.1608) (0.1750,0.1731) (0.1875,0.1854) (0.2000,0.1977) (0.2125,0.2100) (0.2250,0.2222) (0.2375,0.2345) (0.2500,0.2467) (0.2625,0.2589) (0.2750,0.2711) (0.2875,0.2832) (0.3000,0.2954) (0.3125,0.3075) (0.3250,0.3196) (0.3375,0.3317) (0.3500,0.3437) (0.3625,0.3557) (0.3750,0.3677) (0.3875,0.3797) (0.4000,0.3916) (0.4125,0.4035) (0.4250,0.4153) (0.4375,0.4271) (0.4500,0.4389) (0.4625,0.4506) (0.4750,0.4623) (0.4875,0.4739) (0.5000,0.4855) (0.5125,0.4970) (0.5250,0.5085) (0.5375,0.5199) (0.5500,0.5312) (0.5625,0.5425) (0.5750,0.5537) (0.5875,0.5648) (0.6000,0.5758) (0.6125,0.5868) (0.6250,0.5977) (0.6375,0.6084) (0.6500,0.6191) (0.6625,0.6297) (0.6750,0.6402) (0.6875,0.6506) (0.7000,0.6608) (0.7125,0.6710) (0.7250,0.6810) (0.7375,0.6909) (0.7500,0.7007) (0.7625,0.7103) (0.7750,0.7198) (0.7875,0.7292) (0.8000,0.7384) (0.8125,0.7474) (0.8250,0.7563) (0.8375,0.7650) (0.8500,0.7736) (0.8625,0.7820) (0.8750,0.7902) (0.8875,0.7982) (0.9000,0.8061) (0.9125,0.8137) (0.9250,0.8212) (0.9375,0.8285) (0.9500,0.8356) (0.9625,0.8425) (0.9750,0.8492) (0.9875,0.8558) (1.0000,0.8621) (1.0125,0.8682) (1.0250,0.8742) (1.0375,0.8799) (1.0500,0.8855) (1.0625,0.8908) (1.0750,0.8960) (1.0875,0.9010) (1.1000,0.9058) (1.1125,0.9104) (1.1250,0.9148) (1.1375,0.9191) (1.1500,0.9231) (1.1625,0.9270) (1.1750,0.9308) (1.1875,0.9344) (1.2000,0.9378) (1.2125,0.9411) (1.2250,0.9442) (1.2375,0.9472) (1.2500,0.9500) (1.2625,0.9527) (1.2750,0.9553) (1.2875,0.9577) (1.3000,0.9601) (1.3125,0.9623) (1.3250,0.9644) (1.3375,0.9664) (1.3500,0.9682) (1.3625,0.9700) (1.3750,0.9717) (1.3875,0.9733) (1.4000,0.9749) (1.4125,0.9763) (1.4250,0.9777) (1.4375,0.9790) (1.4500,0.9802) (1.4625,0.9813) (1.4750,0.9824) (1.4875,0.9835) (1.5000,0.9844) (1.5125,0.9853) (1.5250,0.9862) (1.5375,0.9870) (1.5500,0.9878) (1.5625,0.9885) (1.5750,0.9892) (1.5875,0.9898) (1.6000,0.9905) (1.6125,0.9910) (1.6250,0.9916) (1.6375,0.9921) (1.6500,0.9925) (1.6625,0.9930) (1.6750,0.9934) (1.6875,0.9938) (1.7000,0.9942) (1.7125,0.9945) (1.7250,0.9949) (1.7375,0.9952) (1.7500,0.9955) (1.7625,0.9958) (1.7750,0.9960) (1.7875,0.9963) (1.8000,0.9965) (1.8125,0.9967) (1.8250,0.9969) (1.8375,0.9971) (1.8500,0.9973) (1.8625,0.9975) (1.8750,0.9976) (1.8875,0.9978) (1.9000,0.9979) (1.9125,0.9980) (1.9250,0.9982) (1.9375,0.9983) (1.9500,0.9984) (1.9625,0.9985) (1.9750,0.9986) (1.9875,0.9987) (2.0000,0.9988) (2.0125,0.9989) (2.0250,0.9989) (2.0375,0.9990) (2.0500,0.9991) (2.0625,0.9991) (2.0750,0.9992) (2.0875,0.9992) (2.1000,0.9993) (2.1125,0.9993) (2.1250,0.9994) (2.1375,0.9994) (2.1500,0.9995) (2.1625,0.9995) (2.1750,0.9996) (2.1875,0.9996) (2.2000,0.9996) (2.2125,0.9996) (2.2250,0.9997) (2.2375,0.9997) (2.2500,0.9997) (2.2625,0.9997) (2.2750,0.9998) (2.2875,0.9998) (2.3000,0.9998) (2.3125,0.9998) (2.3250,0.9998) (2.3375,0.9999) (2.3500,0.9999) (2.3625,0.9999) (2.3750,0.9999) (2.3875,0.9999) (2.4000,0.9999) (2.4125,0.9999) (2.4250,1.0000) (2.4375,1.0000) (2.4500,1.0000) (2.4625,1.0000) (2.4750,1.0000) (2.4875,1.0000) (2.5000,1.0000)};
  \addplot[s2] coordinates {(0.0000,0.0000) (0.0125,0.0109) (0.0250,0.0217) (0.0375,0.0325) (0.0500,0.0432) (0.0625,0.0540) (0.0750,0.0646) (0.0875,0.0753) (0.1000,0.0859) (0.1125,0.0964) (0.1250,0.1069) (0.1375,0.1174) (0.1500,0.1278) (0.1625,0.1382) (0.1750,0.1486) (0.1875,0.1589) (0.2000,0.1691) (0.2125,0.1793) (0.2250,0.1895) (0.2375,0.1996) (0.2500,0.2097) (0.2625,0.2197) (0.2750,0.2297) (0.2875,0.2396) (0.3000,0.2494) (0.3125,0.2592) (0.3250,0.2690) (0.3375,0.2787) (0.3500,0.2883) (0.3625,0.2979) (0.3750,0.3075) (0.3875,0.3169) (0.4000,0.3264) (0.4125,0.3357) (0.4250,0.3450) (0.4375,0.3543) (0.4500,0.3635) (0.4625,0.3726) (0.4750,0.3816) (0.4875,0.3906) (0.5000,0.3996) (0.5125,0.4084) (0.5250,0.4172) (0.5375,0.4260) (0.5500,0.4347) (0.5625,0.4433) (0.5750,0.4518) (0.5875,0.4603) (0.6000,0.4687) (0.6125,0.4770) (0.6250,0.4853) (0.6375,0.4935) (0.6500,0.5016) (0.6625,0.5096) (0.6750,0.5176) (0.6875,0.5255) (0.7000,0.5334) (0.7125,0.5411) (0.7250,0.5488) (0.7375,0.5564) (0.7500,0.5640) (0.7625,0.5714) (0.7750,0.5788) (0.7875,0.5861) (0.8000,0.5934) (0.8125,0.6005) (0.8250,0.6076) (0.8375,0.6146) (0.8500,0.6215) (0.8625,0.6284) (0.8750,0.6352) (0.8875,0.6419) (0.9000,0.6485) (0.9125,0.6550) (0.9250,0.6615) (0.9375,0.6679) (0.9500,0.6742) (0.9625,0.6804) (0.9750,0.6865) (0.9875,0.6926) (1.0000,0.6986) (1.0125,0.7045) (1.0250,0.7104) (1.0375,0.7161) (1.0500,0.7218) (1.0625,0.7274) (1.0750,0.7329) (1.0875,0.7384) (1.1000,0.7438) (1.1125,0.7491) (1.1250,0.7543) (1.1375,0.7594) (1.1500,0.7645) (1.1625,0.7695) (1.1750,0.7744) (1.1875,0.7793) (1.2000,0.7840) (1.2125,0.7887) (1.2250,0.7934) (1.2375,0.7979) (1.2500,0.8024) (1.2625,0.8068) (1.2750,0.8112) (1.2875,0.8154) (1.3000,0.8196) (1.3125,0.8237) (1.3250,0.8278) (1.3375,0.8318) (1.3500,0.8357) (1.3625,0.8396) (1.3750,0.8434) (1.3875,0.8471) (1.4000,0.8508) (1.4125,0.8544) (1.4250,0.8579) (1.4375,0.8614) (1.4500,0.8648) (1.4625,0.8681) (1.4750,0.8714) (1.4875,0.8746) (1.5000,0.8778) (1.5125,0.8809) (1.5250,0.8839) (1.5375,0.8869) (1.5500,0.8899) (1.5625,0.8927) (1.5750,0.8956) (1.5875,0.8983) (1.6000,0.9010) (1.6125,0.9037) (1.6250,0.9063) (1.6375,0.9089) (1.6500,0.9114) (1.6625,0.9138) (1.6750,0.9163) (1.6875,0.9186) (1.7000,0.9209) (1.7125,0.9232) (1.7250,0.9254) (1.7375,0.9276) (1.7500,0.9297) (1.7625,0.9318) (1.7750,0.9339) (1.7875,0.9359) (1.8000,0.9378) (1.8125,0.9397) (1.8250,0.9416) (1.8375,0.9435) (1.8500,0.9453) (1.8625,0.9470) (1.8750,0.9488) (1.8875,0.9504) (1.9000,0.9521) (1.9125,0.9537) (1.9250,0.9553) (1.9375,0.9568) (1.9500,0.9584) (1.9625,0.9598) (1.9750,0.9613) (1.9875,0.9627) (2.0000,0.9641) (2.0125,0.9654) (2.0250,0.9668) (2.0375,0.9681) (2.0500,0.9693) (2.0625,0.9706) (2.0750,0.9718) (2.0875,0.9730) (2.1000,0.9741) (2.1125,0.9753) (2.1250,0.9764) (2.1375,0.9774) (2.1500,0.9785) (2.1625,0.9795) (2.1750,0.9805) (2.1875,0.9815) (2.2000,0.9825) (2.2125,0.9834) (2.2250,0.9844) (2.2375,0.9853) (2.2500,0.9861) (2.2625,0.9870) (2.2750,0.9878) (2.2875,0.9887) (2.3000,0.9895) (2.3125,0.9902) (2.3250,0.9910) (2.3375,0.9917) (2.3500,0.9925) (2.3625,0.9932) (2.3750,0.9939) (2.3875,0.9946) (2.4000,0.9952) (2.4125,0.9959) (2.4250,0.9965) (2.4375,0.9971) (2.4500,0.9977) (2.4625,0.9983) (2.4750,0.9989) (2.4875,0.9995) (2.5000,1.0000)};
  \addplot[s2!55] coordinates {(0.0000,0.0000) (0.0125,0.0069) (0.0250,0.0137) (0.0375,0.0206) (0.0500,0.0274) (0.0625,0.0342) (0.0750,0.0410) (0.0875,0.0478) (0.1000,0.0545) (0.1125,0.0613) (0.1250,0.0680) (0.1375,0.0747) (0.1500,0.0813) (0.1625,0.0880) (0.1750,0.0946) (0.1875,0.1012) (0.2000,0.1078) (0.2125,0.1144) (0.2250,0.1209) (0.2375,0.1275) (0.2500,0.1340) (0.2625,0.1405) (0.2750,0.1470) (0.2875,0.1534) (0.3000,0.1599) (0.3125,0.1663) (0.3250,0.1727) (0.3375,0.1790) (0.3500,0.1854) (0.3625,0.1917) (0.3750,0.1981) (0.3875,0.2044) (0.4000,0.2106) (0.4125,0.2169) (0.4250,0.2231) (0.4375,0.2294) (0.4500,0.2356) (0.4625,0.2418) (0.4750,0.2479) (0.4875,0.2541) (0.5000,0.2602) (0.5125,0.2663) (0.5250,0.2724) (0.5375,0.2784) (0.5500,0.2845) (0.5625,0.2905) (0.5750,0.2965) (0.5875,0.3025) (0.6000,0.3085) (0.6125,0.3144) (0.6250,0.3203) (0.6375,0.3263) (0.6500,0.3321) (0.6625,0.3380) (0.6750,0.3439) (0.6875,0.3497) (0.7000,0.3555) (0.7125,0.3613) (0.7250,0.3671) (0.7375,0.3728) (0.7500,0.3786) (0.7625,0.3843) (0.7750,0.3900) (0.7875,0.3957) (0.8000,0.4013) (0.8125,0.4070) (0.8250,0.4126) (0.8375,0.4182) (0.8500,0.4238) (0.8625,0.4293) (0.8750,0.4349) (0.8875,0.4404) (0.9000,0.4459) (0.9125,0.4514) (0.9250,0.4568) (0.9375,0.4623) (0.9500,0.4677) (0.9625,0.4731) (0.9750,0.4785) (0.9875,0.4839) (1.0000,0.4892) (1.0125,0.4945) (1.0250,0.4998) (1.0375,0.5051) (1.0500,0.5104) (1.0625,0.5157) (1.0750,0.5209) (1.0875,0.5261) (1.1000,0.5313) (1.1125,0.5365) (1.1250,0.5417) (1.1375,0.5468) (1.1500,0.5519) (1.1625,0.5570) (1.1750,0.5621) (1.1875,0.5672) (1.2000,0.5722) (1.2125,0.5772) (1.2250,0.5822) (1.2375,0.5872) (1.2500,0.5922) (1.2625,0.5972) (1.2750,0.6021) (1.2875,0.6070) (1.3000,0.6119) (1.3125,0.6168) (1.3250,0.6216) (1.3375,0.6265) (1.3500,0.6313) (1.3625,0.6361) (1.3750,0.6409) (1.3875,0.6457) (1.4000,0.6504) (1.4125,0.6551) (1.4250,0.6598) (1.4375,0.6645) (1.4500,0.6692) (1.4625,0.6739) (1.4750,0.6785) (1.4875,0.6831) (1.5000,0.6877) (1.5125,0.6923) (1.5250,0.6969) (1.5375,0.7014) (1.5500,0.7059) (1.5625,0.7105) (1.5750,0.7150) (1.5875,0.7194) (1.6000,0.7239) (1.6125,0.7283) (1.6250,0.7327) (1.6375,0.7371) (1.6500,0.7415) (1.6625,0.7459) (1.6750,0.7503) (1.6875,0.7546) (1.7000,0.7589) (1.7125,0.7632) (1.7250,0.7675) (1.7375,0.7717) (1.7500,0.7760) (1.7625,0.7802) (1.7750,0.7844) (1.7875,0.7886) (1.8000,0.7928) (1.8125,0.7969) (1.8250,0.8011) (1.8375,0.8052) (1.8500,0.8093) (1.8625,0.8134) (1.8750,0.8175) (1.8875,0.8215) (1.9000,0.8256) (1.9125,0.8296) (1.9250,0.8336) (1.9375,0.8376) (1.9500,0.8415) (1.9625,0.8455) (1.9750,0.8494) (1.9875,0.8533) (2.0000,0.8572) (2.0125,0.8611) (2.0250,0.8650) (2.0375,0.8688) (2.0500,0.8727) (2.0625,0.8765) (2.0750,0.8803) (2.0875,0.8841) (2.1000,0.8879) (2.1125,0.8916) (2.1250,0.8953) (2.1375,0.8991) (2.1500,0.9028) (2.1625,0.9065) (2.1750,0.9101) (2.1875,0.9138) (2.2000,0.9174) (2.2125,0.9210) (2.2250,0.9246) (2.2375,0.9282) (2.2500,0.9318) (2.2625,0.9354) (2.2750,0.9389) (2.2875,0.9424) (2.3000,0.9459) (2.3125,0.9494) (2.3250,0.9529) (2.3375,0.9564) (2.3500,0.9598) (2.3625,0.9633) (2.3750,0.9667) (2.3875,0.9701) (2.4000,0.9735) (2.4125,0.9768) (2.4250,0.9802) (2.4375,0.9835) (2.4500,0.9869) (2.4625,0.9902) (2.4750,0.9935) (2.4875,0.9967) (2.5000,1.0000)};
  \addplot[black, thick] coordinates {(0.0000,0.0000) (0.0125,0.0050) (0.0250,0.0100) (0.0375,0.0150) (0.0500,0.0200) (0.0625,0.0250) (0.0750,0.0300) (0.0875,0.0350) (0.1000,0.0400) (0.1125,0.0450) (0.1250,0.0500) (0.1375,0.0550) (0.1500,0.0600) (0.1625,0.0650) (0.1750,0.0700) (0.1875,0.0750) (0.2000,0.0800) (0.2125,0.0850) (0.2250,0.0900) (0.2375,0.0950) (0.2500,0.1000) (0.2625,0.1050) (0.2750,0.1100) (0.2875,0.1150) (0.3000,0.1200) (0.3125,0.1250) (0.3250,0.1300) (0.3375,0.1350) (0.3500,0.1400) (0.3625,0.1450) (0.3750,0.1500) (0.3875,0.1550) (0.4000,0.1600) (0.4125,0.1650) (0.4250,0.1700) (0.4375,0.1750) (0.4500,0.1800) (0.4625,0.1850) (0.4750,0.1900) (0.4875,0.1950) (0.5000,0.2000) (0.5125,0.2050) (0.5250,0.2100) (0.5375,0.2150) (0.5500,0.2200) (0.5625,0.2250) (0.5750,0.2300) (0.5875,0.2350) (0.6000,0.2400) (0.6125,0.2450) (0.6250,0.2500) (0.6375,0.2550) (0.6500,0.2600) (0.6625,0.2650) (0.6750,0.2700) (0.6875,0.2750) (0.7000,0.2800) (0.7125,0.2850) (0.7250,0.2900) (0.7375,0.2950) (0.7500,0.3000) (0.7625,0.3050) (0.7750,0.3100) (0.7875,0.3150) (0.8000,0.3200) (0.8125,0.3250) (0.8250,0.3300) (0.8375,0.3350) (0.8500,0.3400) (0.8625,0.3450) (0.8750,0.3500) (0.8875,0.3550) (0.9000,0.3600) (0.9125,0.3650) (0.9250,0.3700) (0.9375,0.3750) (0.9500,0.3800) (0.9625,0.3850) (0.9750,0.3900) (0.9875,0.3950) (1.0000,0.4000) (1.0125,0.4050) (1.0250,0.4100) (1.0375,0.4150) (1.0500,0.4200) (1.0625,0.4250) (1.0750,0.4300) (1.0875,0.4350) (1.1000,0.4400) (1.1125,0.4450) (1.1250,0.4500) (1.1375,0.4550) (1.1500,0.4600) (1.1625,0.4650) (1.1750,0.4700) (1.1875,0.4750) (1.2000,0.4800) (1.2125,0.4850) (1.2250,0.4900) (1.2375,0.4950) (1.2500,0.5000) (1.2625,0.5050) (1.2750,0.5100) (1.2875,0.5150) (1.3000,0.5200) (1.3125,0.5250) (1.3250,0.5300) (1.3375,0.5350) (1.3500,0.5400) (1.3625,0.5450) (1.3750,0.5500) (1.3875,0.5550) (1.4000,0.5600) (1.4125,0.5650) (1.4250,0.5700) (1.4375,0.5750) (1.4500,0.5800) (1.4625,0.5850) (1.4750,0.5900) (1.4875,0.5950) (1.5000,0.6000) (1.5125,0.6050) (1.5250,0.6100) (1.5375,0.6150) (1.5500,0.6200) (1.5625,0.6250) (1.5750,0.6300) (1.5875,0.6350) (1.6000,0.6400) (1.6125,0.6450) (1.6250,0.6500) (1.6375,0.6550) (1.6500,0.6600) (1.6625,0.6650) (1.6750,0.6700) (1.6875,0.6750) (1.7000,0.6800) (1.7125,0.6850) (1.7250,0.6900) (1.7375,0.6950) (1.7500,0.7000) (1.7625,0.7050) (1.7750,0.7100) (1.7875,0.7150) (1.8000,0.7200) (1.8125,0.7250) (1.8250,0.7300) (1.8375,0.7350) (1.8500,0.7400) (1.8625,0.7450) (1.8750,0.7500) (1.8875,0.7550) (1.9000,0.7600) (1.9125,0.7650) (1.9250,0.7700) (1.9375,0.7750) (1.9500,0.7800) (1.9625,0.7850) (1.9750,0.7900) (1.9875,0.7950) (2.0000,0.8000) (2.0125,0.8050) (2.0250,0.8100) (2.0375,0.8150) (2.0500,0.8200) (2.0625,0.8250) (2.0750,0.8300) (2.0875,0.8350) (2.1000,0.8400) (2.1125,0.8450) (2.1250,0.8500) (2.1375,0.8550) (2.1500,0.8600) (2.1625,0.8650) (2.1750,0.8700) (2.1875,0.8750) (2.2000,0.8800) (2.2125,0.8850) (2.2250,0.8900) (2.2375,0.8950) (2.2500,0.9000) (2.2625,0.9050) (2.2750,0.9100) (2.2875,0.9150) (2.3000,0.9200) (2.3125,0.9250) (2.3250,0.9300) (2.3375,0.9350) (2.3500,0.9400) (2.3625,0.9450) (2.3750,0.9500) (2.3875,0.9550) (2.4000,0.9600) (2.4125,0.9650) (2.4250,0.9700) (2.4375,0.9750) (2.4500,0.9800) (2.4625,0.9850) (2.4750,0.9900) (2.4875,0.9950) (2.5000,1.0000)};
  \addplot[s3!55] coordinates {(0.0000,0.0000) (0.0125,0.0033) (0.0250,0.0065) (0.0375,0.0098) (0.0500,0.0131) (0.0625,0.0165) (0.0750,0.0198) (0.0875,0.0232) (0.1000,0.0265) (0.1125,0.0299) (0.1250,0.0333) (0.1375,0.0367) (0.1500,0.0402) (0.1625,0.0436) (0.1750,0.0471) (0.1875,0.0506) (0.2000,0.0541) (0.2125,0.0576) (0.2250,0.0611) (0.2375,0.0646) (0.2500,0.0682) (0.2625,0.0718) (0.2750,0.0754) (0.2875,0.0790) (0.3000,0.0826) (0.3125,0.0862) (0.3250,0.0899) (0.3375,0.0935) (0.3500,0.0972) (0.3625,0.1009) (0.3750,0.1047) (0.3875,0.1084) (0.4000,0.1121) (0.4125,0.1159) (0.4250,0.1197) (0.4375,0.1235) (0.4500,0.1273) (0.4625,0.1312) (0.4750,0.1350) (0.4875,0.1389) (0.5000,0.1428) (0.5125,0.1467) (0.5250,0.1506) (0.5375,0.1545) (0.5500,0.1585) (0.5625,0.1624) (0.5750,0.1664) (0.5875,0.1704) (0.6000,0.1744) (0.6125,0.1785) (0.6250,0.1825) (0.6375,0.1866) (0.6500,0.1907) (0.6625,0.1948) (0.6750,0.1989) (0.6875,0.2031) (0.7000,0.2072) (0.7125,0.2114) (0.7250,0.2156) (0.7375,0.2198) (0.7500,0.2240) (0.7625,0.2283) (0.7750,0.2325) (0.7875,0.2368) (0.8000,0.2411) (0.8125,0.2454) (0.8250,0.2497) (0.8375,0.2541) (0.8500,0.2585) (0.8625,0.2629) (0.8750,0.2673) (0.8875,0.2717) (0.9000,0.2761) (0.9125,0.2806) (0.9250,0.2850) (0.9375,0.2895) (0.9500,0.2941) (0.9625,0.2986) (0.9750,0.3031) (0.9875,0.3077) (1.0000,0.3123) (1.0125,0.3169) (1.0250,0.3215) (1.0375,0.3261) (1.0500,0.3308) (1.0625,0.3355) (1.0750,0.3402) (1.0875,0.3449) (1.1000,0.3496) (1.1125,0.3543) (1.1250,0.3591) (1.1375,0.3639) (1.1500,0.3687) (1.1625,0.3735) (1.1750,0.3784) (1.1875,0.3832) (1.2000,0.3881) (1.2125,0.3930) (1.2250,0.3979) (1.2375,0.4028) (1.2500,0.4078) (1.2625,0.4128) (1.2750,0.4178) (1.2875,0.4228) (1.3000,0.4278) (1.3125,0.4328) (1.3250,0.4379) (1.3375,0.4430) (1.3500,0.4481) (1.3625,0.4532) (1.3750,0.4583) (1.3875,0.4635) (1.4000,0.4687) (1.4125,0.4739) (1.4250,0.4791) (1.4375,0.4843) (1.4500,0.4896) (1.4625,0.4949) (1.4750,0.5002) (1.4875,0.5055) (1.5000,0.5108) (1.5125,0.5161) (1.5250,0.5215) (1.5375,0.5269) (1.5500,0.5323) (1.5625,0.5377) (1.5750,0.5432) (1.5875,0.5486) (1.6000,0.5541) (1.6125,0.5596) (1.6250,0.5651) (1.6375,0.5707) (1.6500,0.5762) (1.6625,0.5818) (1.6750,0.5874) (1.6875,0.5930) (1.7000,0.5987) (1.7125,0.6043) (1.7250,0.6100) (1.7375,0.6157) (1.7500,0.6214) (1.7625,0.6272) (1.7750,0.6329) (1.7875,0.6387) (1.8000,0.6445) (1.8125,0.6503) (1.8250,0.6561) (1.8375,0.6620) (1.8500,0.6679) (1.8625,0.6737) (1.8750,0.6797) (1.8875,0.6856) (1.9000,0.6915) (1.9125,0.6975) (1.9250,0.7035) (1.9375,0.7095) (1.9500,0.7155) (1.9625,0.7216) (1.9750,0.7276) (1.9875,0.7337) (2.0000,0.7398) (2.0125,0.7459) (2.0250,0.7521) (2.0375,0.7582) (2.0500,0.7644) (2.0625,0.7706) (2.0750,0.7769) (2.0875,0.7831) (2.1000,0.7894) (2.1125,0.7956) (2.1250,0.8019) (2.1375,0.8083) (2.1500,0.8146) (2.1625,0.8210) (2.1750,0.8273) (2.1875,0.8337) (2.2000,0.8401) (2.2125,0.8466) (2.2250,0.8530) (2.2375,0.8595) (2.2500,0.8660) (2.2625,0.8725) (2.2750,0.8791) (2.2875,0.8856) (2.3000,0.8922) (2.3125,0.8988) (2.3250,0.9054) (2.3375,0.9120) (2.3500,0.9187) (2.3625,0.9253) (2.3750,0.9320) (2.3875,0.9387) (2.4000,0.9455) (2.4125,0.9522) (2.4250,0.9590) (2.4375,0.9658) (2.4500,0.9726) (2.4625,0.9794) (2.4750,0.9863) (2.4875,0.9931) (2.5000,1.0000)};
  \addplot[s3] coordinates {(0.0000,0.0000) (0.0125,0.0005) (0.0250,0.0011) (0.0375,0.0017) (0.0500,0.0023) (0.0625,0.0029) (0.0750,0.0035) (0.0875,0.0041) (0.1000,0.0048) (0.1125,0.0054) (0.1250,0.0061) (0.1375,0.0068) (0.1500,0.0075) (0.1625,0.0083) (0.1750,0.0090) (0.1875,0.0098) (0.2000,0.0105) (0.2125,0.0113) (0.2250,0.0122) (0.2375,0.0130) (0.2500,0.0139) (0.2625,0.0147) (0.2750,0.0156) (0.2875,0.0166) (0.3000,0.0175) (0.3125,0.0185) (0.3250,0.0195) (0.3375,0.0205) (0.3500,0.0215) (0.3625,0.0226) (0.3750,0.0236) (0.3875,0.0247) (0.4000,0.0259) (0.4125,0.0270) (0.4250,0.0282) (0.4375,0.0294) (0.4500,0.0307) (0.4625,0.0319) (0.4750,0.0332) (0.4875,0.0346) (0.5000,0.0359) (0.5125,0.0373) (0.5250,0.0387) (0.5375,0.0402) (0.5500,0.0416) (0.5625,0.0432) (0.5750,0.0447) (0.5875,0.0463) (0.6000,0.0479) (0.6125,0.0496) (0.6250,0.0512) (0.6375,0.0530) (0.6500,0.0547) (0.6625,0.0565) (0.6750,0.0584) (0.6875,0.0603) (0.7000,0.0622) (0.7125,0.0641) (0.7250,0.0661) (0.7375,0.0682) (0.7500,0.0703) (0.7625,0.0724) (0.7750,0.0746) (0.7875,0.0768) (0.8000,0.0791) (0.8125,0.0814) (0.8250,0.0837) (0.8375,0.0862) (0.8500,0.0886) (0.8625,0.0911) (0.8750,0.0937) (0.8875,0.0963) (0.9000,0.0990) (0.9125,0.1017) (0.9250,0.1044) (0.9375,0.1073) (0.9500,0.1101) (0.9625,0.1131) (0.9750,0.1161) (0.9875,0.1191) (1.0000,0.1222) (1.0125,0.1254) (1.0250,0.1286) (1.0375,0.1319) (1.0500,0.1352) (1.0625,0.1386) (1.0750,0.1421) (1.0875,0.1456) (1.1000,0.1492) (1.1125,0.1529) (1.1250,0.1566) (1.1375,0.1604) (1.1500,0.1643) (1.1625,0.1682) (1.1750,0.1722) (1.1875,0.1763) (1.2000,0.1804) (1.2125,0.1846) (1.2250,0.1888) (1.2375,0.1932) (1.2500,0.1976) (1.2625,0.2021) (1.2750,0.2066) (1.2875,0.2113) (1.3000,0.2160) (1.3125,0.2207) (1.3250,0.2256) (1.3375,0.2305) (1.3500,0.2355) (1.3625,0.2406) (1.3750,0.2457) (1.3875,0.2509) (1.4000,0.2562) (1.4125,0.2616) (1.4250,0.2671) (1.4375,0.2726) (1.4500,0.2782) (1.4625,0.2839) (1.4750,0.2896) (1.4875,0.2955) (1.5000,0.3014) (1.5125,0.3074) (1.5250,0.3135) (1.5375,0.3196) (1.5500,0.3258) (1.5625,0.3321) (1.5750,0.3385) (1.5875,0.3450) (1.6000,0.3515) (1.6125,0.3581) (1.6250,0.3648) (1.6375,0.3716) (1.6500,0.3785) (1.6625,0.3854) (1.6750,0.3924) (1.6875,0.3995) (1.7000,0.4066) (1.7125,0.4139) (1.7250,0.4212) (1.7375,0.4286) (1.7500,0.4360) (1.7625,0.4436) (1.7750,0.4512) (1.7875,0.4589) (1.8000,0.4666) (1.8125,0.4745) (1.8250,0.4824) (1.8375,0.4904) (1.8500,0.4984) (1.8625,0.5065) (1.8750,0.5147) (1.8875,0.5230) (1.9000,0.5313) (1.9125,0.5397) (1.9250,0.5482) (1.9375,0.5567) (1.9500,0.5653) (1.9625,0.5740) (1.9750,0.5828) (1.9875,0.5916) (2.0000,0.6004) (2.0125,0.6094) (2.0250,0.6184) (2.0375,0.6274) (2.0500,0.6365) (2.0625,0.6457) (2.0750,0.6550) (2.0875,0.6643) (2.1000,0.6736) (2.1125,0.6831) (2.1250,0.6925) (2.1375,0.7021) (2.1500,0.7117) (2.1625,0.7213) (2.1750,0.7310) (2.1875,0.7408) (2.2000,0.7506) (2.2125,0.7604) (2.2250,0.7703) (2.2375,0.7803) (2.2500,0.7903) (2.2625,0.8004) (2.2750,0.8105) (2.2875,0.8207) (2.3000,0.8309) (2.3125,0.8411) (2.3250,0.8514) (2.3375,0.8618) (2.3500,0.8722) (2.3625,0.8826) (2.3750,0.8931) (2.3875,0.9036) (2.4000,0.9141) (2.4125,0.9247) (2.4250,0.9354) (2.4375,0.9460) (2.4500,0.9568) (2.4625,0.9675) (2.4750,0.9783) (2.4875,0.9891) (2.5000,1.0000)};
  \addplot[s4] coordinates {(0.0000,0.0000) (0.0125,0.0000) (0.0250,0.0000) (0.0375,0.0000) (0.0500,0.0000) (0.0625,0.0000) (0.0750,0.0000) (0.0875,0.0001) (0.1000,0.0001) (0.1125,0.0001) (0.1250,0.0001) (0.1375,0.0001) (0.1500,0.0001) (0.1625,0.0001) (0.1750,0.0002) (0.1875,0.0002) (0.2000,0.0002) (0.2125,0.0002) (0.2250,0.0002) (0.2375,0.0003) (0.2500,0.0003) (0.2625,0.0003) (0.2750,0.0003) (0.2875,0.0004) (0.3000,0.0004) (0.3125,0.0004) (0.3250,0.0004) (0.3375,0.0005) (0.3500,0.0005) (0.3625,0.0006) (0.3750,0.0006) (0.3875,0.0007) (0.4000,0.0007) (0.4125,0.0008) (0.4250,0.0008) (0.4375,0.0009) (0.4500,0.0009) (0.4625,0.0010) (0.4750,0.0011) (0.4875,0.0011) (0.5000,0.0012) (0.5125,0.0013) (0.5250,0.0014) (0.5375,0.0015) (0.5500,0.0016) (0.5625,0.0017) (0.5750,0.0018) (0.5875,0.0020) (0.6000,0.0021) (0.6125,0.0022) (0.6250,0.0024) (0.6375,0.0025) (0.6500,0.0027) (0.6625,0.0029) (0.6750,0.0031) (0.6875,0.0033) (0.7000,0.0035) (0.7125,0.0037) (0.7250,0.0040) (0.7375,0.0042) (0.7500,0.0045) (0.7625,0.0048) (0.7750,0.0051) (0.7875,0.0055) (0.8000,0.0058) (0.8125,0.0062) (0.8250,0.0066) (0.8375,0.0070) (0.8500,0.0075) (0.8625,0.0079) (0.8750,0.0084) (0.8875,0.0090) (0.9000,0.0095) (0.9125,0.0102) (0.9250,0.0108) (0.9375,0.0115) (0.9500,0.0122) (0.9625,0.0130) (0.9750,0.0138) (0.9875,0.0147) (1.0000,0.0156) (1.0125,0.0165) (1.0250,0.0176) (1.0375,0.0187) (1.0500,0.0198) (1.0625,0.0210) (1.0750,0.0223) (1.0875,0.0237) (1.1000,0.0251) (1.1125,0.0267) (1.1250,0.0283) (1.1375,0.0300) (1.1500,0.0318) (1.1625,0.0336) (1.1750,0.0356) (1.1875,0.0377) (1.2000,0.0399) (1.2125,0.0423) (1.2250,0.0447) (1.2375,0.0473) (1.2500,0.0500) (1.2625,0.0528) (1.2750,0.0558) (1.2875,0.0589) (1.3000,0.0622) (1.3125,0.0656) (1.3250,0.0692) (1.3375,0.0730) (1.3500,0.0769) (1.3625,0.0809) (1.3750,0.0852) (1.3875,0.0896) (1.4000,0.0942) (1.4125,0.0990) (1.4250,0.1040) (1.4375,0.1092) (1.4500,0.1145) (1.4625,0.1201) (1.4750,0.1258) (1.4875,0.1318) (1.5000,0.1379) (1.5125,0.1442) (1.5250,0.1508) (1.5375,0.1575) (1.5500,0.1644) (1.5625,0.1715) (1.5750,0.1788) (1.5875,0.1863) (1.6000,0.1939) (1.6125,0.2018) (1.6250,0.2098) (1.6375,0.2180) (1.6500,0.2264) (1.6625,0.2350) (1.6750,0.2437) (1.6875,0.2526) (1.7000,0.2616) (1.7125,0.2708) (1.7250,0.2802) (1.7375,0.2897) (1.7500,0.2993) (1.7625,0.3091) (1.7750,0.3190) (1.7875,0.3290) (1.8000,0.3392) (1.8125,0.3494) (1.8250,0.3598) (1.8375,0.3703) (1.8500,0.3809) (1.8625,0.3916) (1.8750,0.4023) (1.8875,0.4132) (1.9000,0.4242) (1.9125,0.4352) (1.9250,0.4463) (1.9375,0.4575) (1.9500,0.4688) (1.9625,0.4801) (1.9750,0.4915) (1.9875,0.5030) (2.0000,0.5145) (2.0125,0.5261) (2.0250,0.5377) (2.0375,0.5494) (2.0500,0.5611) (2.0625,0.5729) (2.0750,0.5847) (2.0875,0.5965) (2.1000,0.6084) (2.1125,0.6203) (2.1250,0.6323) (2.1375,0.6443) (2.1500,0.6563) (2.1625,0.6683) (2.1750,0.6804) (2.1875,0.6925) (2.2000,0.7046) (2.2125,0.7168) (2.2250,0.7289) (2.2375,0.7411) (2.2500,0.7533) (2.2625,0.7655) (2.2750,0.7778) (2.2875,0.7900) (2.3000,0.8023) (2.3125,0.8146) (2.3250,0.8269) (2.3375,0.8392) (2.3500,0.8515) (2.3625,0.8638) (2.3750,0.8762) (2.3875,0.8885) (2.4000,0.9009) (2.4125,0.9132) (2.4250,0.9256) (2.4375,0.9380) (2.4500,0.9504) (2.4625,0.9628) (2.4750,0.9752) (2.4875,0.9876) (2.5000,1.0000)};
  \addplot[densely dashed, s4] coordinates {(0.0000,0.0000) (1.5000,0.0000) (2.5000,1.0000)};
  \node[font=\scriptsize, anchor=south west] at (axis cs:2.33,1.005) {$c=(2.5,1)$};
  \addplot[only marks, mark=*, mark size=1.4pt, black] coordinates {(0,0) (2.5,1)};
\end{axis}
\end{tikzpicture}
\caption{Award paths of the symmetric log-exponential family for claims $(2.5,1)$. The upper-left dashed path is $\CEA$ ($\theta=-\infty$) and the lower-right dashed path is $\CEL$ ($\theta=+\infty$). From top to bottom the solid paths correspond to $\theta=-5,-2,-0.5,0,0.5,2,5$. The straight path is $\PROP$ ($\theta=0$).}\label{rtheta}
\end{figure}

The parameter ranks the symmetric family in the Lorenz order: for $\theta\le\theta'$, the rule $r^\theta$ Lorenz dominates $r^{\theta'}$ (Proposition~\ref{prop:lorenz}), that is, at every problem and for every $k$, the sum of the $k$ smallest awards chosen by $r^\theta$ is at least the sum of the $k$ smallest awards chosen by $r^{\theta'}$. The family is a chain in this order, with CEA at $\theta=-\infty$ as its maximal and CEL at $\theta=+\infty$ as its minimal member. For two claimants the comparison can be read off Figure~\ref{rtheta}: the award paths do not cross, and the path belonging to the smaller parameter is everywhere closer to the diagonal.

This ranking gives $\theta$ an interpretation in both applications. In taxation it is progressivity: awards are post-tax amounts, so a smaller parameter distributes them more evenly and the induced equal-sacrifice schedule is more progressive. In auto-deleveraging, $\theta$ measures risk sensitivity. In our companion paper \citep{SchlegelMamageishvili2026}, we show that platform risk is weakly increasing in $\theta$.

The log-exponential form arises from combining the two path-independence requirements discussed above with bilateral consistency. Bilateral consistency requires that each claimant in a pair keep their award when the rule is applied to the pair alone, with their claims unchanged and their combined award as the endowment. The two composition axioms require path independence as the endowment is revised down or up. Composition down is path independence in awards, and composition up is path independence in losses. If a rule divides an endowment $E$ into awards $x_i$, its dual divides the deficit $\sum_i c_i-E$ into losses $c_i-x_i$. Because duality exchanges the two composition axioms and preserves bilateral consistency, the class of rules satisfying all three axioms is closed under duality. Young's characterisation imposes composition down. Imposing composition up as well makes the characterisation apply equally to the division of awards and of losses.

Bilateral consistency and the two composition axioms allow claimants with the same claim to be treated differently. Creditors may have different seniority, and accounts with the same position size may have different equity. Theorem~\ref{thm:main} shows how such differences can enter the rule. The claimants fall into ordered priority classes, and a class receives an award only after all earlier classes have been paid in full. Within a priority class, the rule can still distinguish among claimants. A class of at least three claimants uses either a log-exponential rule with possibly different parameters $\theta_i$ and a common budget parameter $\lambda$, or weighted CEA or weighted CEL with possibly different weights. A class of two admits any of the bi-compositional rules of \citet{Chambers2006} for which neither claimant has absolute priority. Continuity follows from the three axioms. Adding \emph{equal treatment of equals} forces a single priority class and selects exactly the symmetric family described above (Corollary~\ref{cor:symmetric}).

Among the rules we characterise, merging-proofness selects the symmetric rules with $\theta\in[-\infty,0]$, splitting-proofness selects the symmetric rules with $\theta\in[0,+\infty]$, and the proportional rule is the only rule satisfying both properties. Invariance under claims truncation selects priority compositions of weighted CEA rules, and minimal rights first selects priority compositions of weighted CEL rules (Section~\ref{sec:compat}).

The main step of the proof concerns the strictly interior symmetric rules. For such a rule, every positive claim is partially satisfied at every interior endowment. At a fixed parameter the gain map sends a claim $c_0$ to its award $x$, and the loss map sends it to $c_0-x$. Young's equal-sacrifice theorem provides a utility $U$ that represents the gain maps additively. Applying the same argument to the dual rule provides a utility $V$ for the loss maps. Compatibility of these two utility representations yields a functional equation which we call the bridge equation. If both utilities are continuously differentiable with positive derivatives, differentiating the bridge equation gives
\[
\frac{U'(c_0)}{U'(x)}+\frac{V'(c_0)}{V'(c_0-x)}=1,\qquad 0<x<c_0,
\]
where $U$ is the gain-side utility and $V$ is the dual loss-side utility. We justify the differentiation by showing that the composition axioms imply strict concavity and $C^1$ regularity with positive derivatives. The only solutions are the log-exponential utilities. If strict interiority fails, that is, if some positive claim is fully paid or receives nothing at an interior endowment, iterating the gain or loss maps forces CEA or CEL. The same boundary argument yields weighted CEA or weighted CEL in the asymmetric case. In the strictly interior asymmetric case, each claimant's gain maps define a symmetric rule on an auxiliary population. The symmetric interior argument identifies the parameter $\theta_i$ of this rule, and the composition axioms imply that all claimants use the same $\lambda$.

\subsection{Related literature}\label{sec:literature}

\emph{One composition versus two.} \citet{Young1988} characterised the \emph{equal-sacrifice} rules using one composition axiom together with ETE, continuity (CONT), bilateral consistency (CONS), strict resource monotonicity and strict order preservation. Within this class, scale invariance selects the isoelastic (power) utilities. Young's composition axiom is our composition down (CD): equal-sacrifice rules satisfy ETE, CONS and CD, but generically violate composition up (CU). Example~3 in \cite{Moulin2000} uses $u(x)=-1/x$ to illustrate this. Our symmetric result adds the second composition axiom and drops both the strictness conditions and continuity: the equal-sacrifice class then reduces to the one-parameter log-exponential subfamily $u_\theta=\log\varphi_\theta$, together with the boundary rules CEA and CEL which are excluded by strictness.

Without the strictness conditions, \citet{ChambersMorenoTernero2017} obtain the generalised equal-sacrifice rules. Such a rule is described by brackets on the claims axis and a continuous increasing sacrifice function on every nondegenerate bracket. Claims below the active bracket are fully honoured, claimants in the active bracket equalise sacrifice subject to the bracket floor, and larger claims are capped at a common award.

\emph{Moulin's characterisation and open problem.} \citet{Moulin2000} studies scale-invariant rules under our axioms (his Upper and Lower Composition are our CD and CU). In a priority class of at least three claimants, the solutions are proportional, weighted-gains and weighted-losses rules,\footnote{We use \citet{Moulin2000}'s terminology here; the weighted-gains and the weighted-losses rules are the weighted versions of CEA and of CEL, respectively.} with further possibilities in two-claimant classes. Exactly three rules satisfying ETE remain: PROP, CEA and CEL, which are precisely the scale-invariant members of our family. Dropping scale invariance turns these three isolated rules into the connected family $\{r^\theta\}$. Moulin poses the problem of characterising the rules without scale invariance and gives one non-homogeneous example, his Example~4. In our notation its awards satisfy $x_2/c_2=(e^{x_1}-1)/(e^{c_1}-1)$ for positive claims. This is $r^\theta$ with $\theta=(1,0)$. It combines the $\theta=1$ gain map $\varphi_1^{-1}(\lambda\varphi_1(\cdot))$ for claimant~1 with the proportional gain map $c_2\mapsto\lambda c_2$ for claimant~2. Theorem~\ref{thm:main} resolves the problem without scale invariance, and Corollary~\ref{cor:symmetric} gives its symmetric case. In the homogeneous case, the shorter proof of \citet{Thomson2013} uses scale invariance to make the two-claimant paths piecewise linear and determines which of them extend to more claimants by counting kinks. The log-exponential paths for finite $\theta$ are smooth, so this argument does not apply here.

\citet{Chambers2006} studies the corresponding fixed-population problem: CD and CU alone characterise the difference rules. Every $r^\theta$ is a difference rule, with paths traced by our gain maps, and Corollary~\ref{cor:symmetric} identifies which difference rules remain after adding ETE and CONS. In the variable-population model, Chambers introduced the asymmetric logarithmic-proportional rules, which give consistent extensions of Moulin's Example~4. Since no symmetric non-homogeneous rule satisfying CONS, CD and CU was known, he conjectured that these axioms together with symmetry characterise PROP, CEA and CEL. The family $r^\theta$ refutes this conjecture: every member with $\theta\in\R\setminus\{0\}$ satisfies ETE, is consistent, bi-compositional and non-homogeneous.\footnote{We impose ETE rather than anonymity throughout. For consistent rules the choice is immaterial: ETE and consistency together imply anonymity \citep[Lemma~3]{ChambersThomson2002}, and Remark~\ref{rem:samepair} gives the corresponding argument with three potential agents under endowment monotonicity and bilateral consistency.}

\emph{The family outside the rationing literature.}
The rules $r^\theta$ are not new as mathematical objects. As we discuss in Remark~\ref{rem:copula}, symmetric strictly interior rules give rise to copulas. The equal-sacrifice representation in awards makes the copula Archimedean, and the representation in losses gives radial symmetry. These properties select the Frank family on each fixed interval \citep[Theorem~4.1]{Frank1979}. Our contribution is the axiomatic selection of the division rules. Frank's general classification also allows ordinal sums with interior idempotents on a fixed interval. In the claims model the same rule must apply on the whole half-line. Our characterisation gives a log-exponential parameter on the claim scale that is independent of the reference claim $\bar c_0$. It also identifies CEA and CEL when strict interiority fails.

\emph{Transversal families.} Several prominent one-parameter interpolations between the focal rules are not bi-compositional and meet our family only at shared endpoints. The non-endpoint members of the Talmud rule and the TAL family \citep{MorenoTerneroVillar2006} connect CEA and CEL but violate at least one composition axiom. Thomson's compromise families between PROP and CEA \citep{Thomson2015compromise} use claims-dependent weighted averages and cone partitions, may violate ETE, and intersect with our symmetric family only at PROP and CEA. On the asymmetric side, \citet{Stovall2014} characterises monotone-path rules through collective rationality, while \citet{MoulinStong2002} obtain the standard-of-gains and standard-of-losses families in the probabilistic discrete model from one composition axiom and monotonicity. In the deterministic discrete model, \citet{Moulin2000} shows that CONS, CD and CU leave only the priority rules. Thus, the continuum characterised here is a divisible-good phenomenon. See \cite{Thomson2003,Thomson2015,Thomson2019} for surveys of the model.

\section{Model, preliminaries and main results}\label{sec:setting}

\subsection{Claims problems and rules}

Let $N^\ast$ be a set of at least three {\bf potential agents}. A {\bf claims problem} is a triple $(N,c,E)$ with a finite, nonempty set $N\subseteq N^\ast$ of agents, a claims vector $c\in\R_+^N$, with component $c_i$ for each agent $i$, and an endowment $0\le E\le\norm{c}$ where $\norm{c}$ is the sum of claims. More generally, for a vector $v\in\R^N$, we write $\norm{v}:=\norm{v}_1:=\sum_{i\in N}|v_i|$ and $\norm{v}_\infty:=\max_{i\in N}|v_i|$.   An {\bf allocation} for $(N,c,E)$ is a vector $x\in\R_+^N$ with $0\le x_i\le c_i$ for each $i$ and $\sum_{i\in N}x_i=E$. A {\bf rule} $r$ assigns an allocation $r(N,c,E)$ to each problem. The component $x_i=r_i(N,c,E)$ is the {\bf award} of agent $i$, $c_i-x_i$ is the agent's {\bf loss}, and $\norm{c}-E$ is the {\bf deficit} of the problem. For a rule $r$ and claims $c$ on $N$, the {\bf award path} is the map $E\mapsto r(N,c,E)$. An endowment is {\bf interior} if $0<E<\norm{c}$. A rule is {\bf strictly interior} if, at every interior endowment, every positive claim is partially satisfied, i.e.~for each problem $(N,c,E)$ with $0<E<\norm{c}$, we have $0<r_i(N,c,E)<c_i$ whenever $c_i>0$.

For two-claimant problems we use the notation $r^{ij}(c,E):=r(\{i,j\},c,E)$, in coordinates $(i,j)$. We suppress $N$ where no confusion arises. For $A\subseteq N$, $c_A$ denotes the restriction of $c$ to $A$, and similarly for other vectors.  The restriction of a rule to $B\subseteq N^\ast$ is its restriction to all problems with $N\subseteq B$. We write $c_0$ for a generic individual claim.

\subsection{Axioms}
We use the following three axioms throughout. See \cite{Thomson2019} and the introduction for extensive motivation.
\begin{description}[leftmargin=2.2em, style=nextline]
    \item[Composition down (CD).] Applying the rule sequentially along decreasing endowments shall lead to the same allocation as applying the rule to the smallest endowment immediately: for $0 \le E' \le E \le \norm{c}$,
  $ r(c, E') = r\big(r(c,E),\, E'\big). $
  \item[Composition up (CU).] Applying the rule sequentially along increasing endowments shall lead to the same allocation as applying the rule to the largest endowment immediately: for $0 \le E \le E' \le \norm{c}$,
  $ r(c, E') = r(c,E) + r\big(c - r(c,E),\, E' - E\big). $
  \item[Bilateral consistency (CONS).] Suppose the rule has been applied to a problem and two of the claimants are singled out. Consider the two-claimant problem in which these two have the same claims as before and the endowment is the sum of the two amounts they were awarded. Consistency requires that the rule award each of them the same amount in this problem as in the original one: for every problem, every pair $\{i,j\} \subseteq N$, and $x = r(N,c,E)$,
  $ (x_i, x_j) = r\big(\{i,j\}, (c_i,c_j),\, x_i + x_j\big). $

\end{description}
A rule satisfying CD and CU is {\bf bi-compositional}. 
Consistency also has a group-wise version, in which an arbitrary subgroup, and not only a pair, is singled out. We impose the bilateral version only. Several proofs use consistency for arbitrary subgroups. Under either composition axiom, bilateral consistency implies this stronger property (Lemma~\ref{lem:fullcons}).
\begin{description}[leftmargin=2.2em, style=nextline]
  \item[Consistency over subgroups (CONS$^\ast$).] For every problem, every nonempty $A \subseteq N$, and $x = r(N,c,E)$,
  $ x_A = r\big(A, c_A,\, \textstyle\sum_{i \in A} x_i\big). $
\end{description}

 CONS, CD and CU imply the following continuity property. Proposition~\ref{prop:auto-cont} gives the Lipschitz bound.
\begin{description}[leftmargin=2.2em, style=nextline]
  \item[Continuity (CONT).] Small changes in the data of a problem should not lead to large changes in awards: for each fixed population $N$, $r(N,\cdot,\cdot)$ is jointly continuous in $(c,E)$.
\end{description}

In Corollaries~\ref{cor:symmetric} and \ref{cor:moulin} we combine these axioms with two other axioms respectively:
\begin{description}[leftmargin=2.2em, style=nextline]
\item[Equal treatment of equals (ETE).] Agents with the same claims shall receive the same awards, $c_i = c_j \implies r_i(c,E) = r_j(c,E)$.
\item[Homogeneity.] Scaling the data of the problem by $\lambda>0$ scales awards by $\lambda$, $r(\lambda c,\lambda E)=\lambda r(c,E).$
\end{description}

\subsection{\texorpdfstring{The log-exponential rules $r^\theta$}{The log-exponential rules r-theta}}

For a scalar $\theta\in\R$, define
\[
\varphi_\theta(x) := \frac{e^{\theta x} - 1}{\theta} \quad (\theta \neq 0), \qquad \varphi_0(x) := x,
\]
a strictly increasing continuous bijection of $[0,\infty)$ onto $[0,\infty)$ for $\theta\ge0$, and onto $[0,1/|\theta|)$ for $\theta<0$. Note $1+\theta\varphi_\theta(x)=e^{\theta x}$, so $\varphi_{\theta}$ satisfies the Cauchy-type equation $1+\theta\varphi_\theta(a+b)=(1+\theta\varphi_\theta(a))(1+\theta\varphi_\theta(b))$.

For a set of potential agents $N^{\ast}$, a parameter-profile $\theta=(\theta_i)_{i\in N^\ast}\in\R^{N^\ast}$ specifies a {\bf log-exponential rule} $r^\theta$ on $N^{\ast}$ by
\begin{equation}\label{eq:rtheta}
 r_i^\theta(N,c,E)=\varphi_{\theta_i}^{-1}\big(\lambda\varphi_{\theta_i}(c_i)\big),
 \qquad \lambda\in[0,1]\text{ determined by } \sum_{j\in N}\varphi_{\theta_j}^{-1}\big(\lambda\varphi_{\theta_j}(c_j)\big)=E.
\end{equation}
If $c\ne0$, the sum is continuous and strictly increasing in $\lambda$ from $0$ to $\norm{c}$, so the endowment determines $\lambda$ uniquely. If $c=0$, every $\lambda\in[0,1]$ gives the zero allocation. At interior endowments the rule equalises sacrifice across positive claims for $u_{\theta_i}=\log\varphi_{\theta_i}$ in the sense of \citet{Young1988}.

We call the log-exponential rule with a constant parameter across agents a {\bf symmetric log-exponential rule}. For a scalar $\theta\in\R$, we write $r^\theta$ for this rule with common parameter $\theta$. At $\theta=0$ we obtain the proportional rule $\PROP$, which awards $Ec_i/\norm{c}$ when $\norm{c}>0$. We write $r^{-\infty}:=\CEA$ and $r^{+\infty}:=\CEL$. Proposition~\ref{prop:endpoints} justifies the notation at the endpoints.

The scalar parameter ranks the symmetric family in the Lorenz order. For $x,y\in\R^N$ with $\sum_i x_i=\sum_i y_i$, we say that $x$ {\bf Lorenz dominates} $y$ if for each $k\le|N|$ the sum of the $k$ smallest coordinates of $x$ is at least the sum of the $k$ smallest coordinates of $y$; domination is {\bf strict} if one of these inequalities is strict. A rule Lorenz dominates another rule if its awards vector does so at every problem.

\begin{proposition}[Lorenz ranking]\label{prop:lorenz}
Let $-\infty\le\theta\le\theta'\le+\infty$. Then $r^\theta$ Lorenz dominates $r^{\theta'}$. If $\theta<\theta'$, the domination is strict at every problem with an interior endowment and at least two distinct positive claims. 
\end{proposition}

The proof is in Appendix~\ref{app:Lorenz}. Its two-claimant case rests on the fact that $u_\theta\circ u_{\theta'}^{-1}$ is concave for $\theta<\theta'$, and the passage to arbitrary populations on the lifting theorem for Lorenz comparisons of \citet{Thomson2012lorenz}. %

\begin{figure}
\centering
\includegraphics[width=0.9\textwidth]{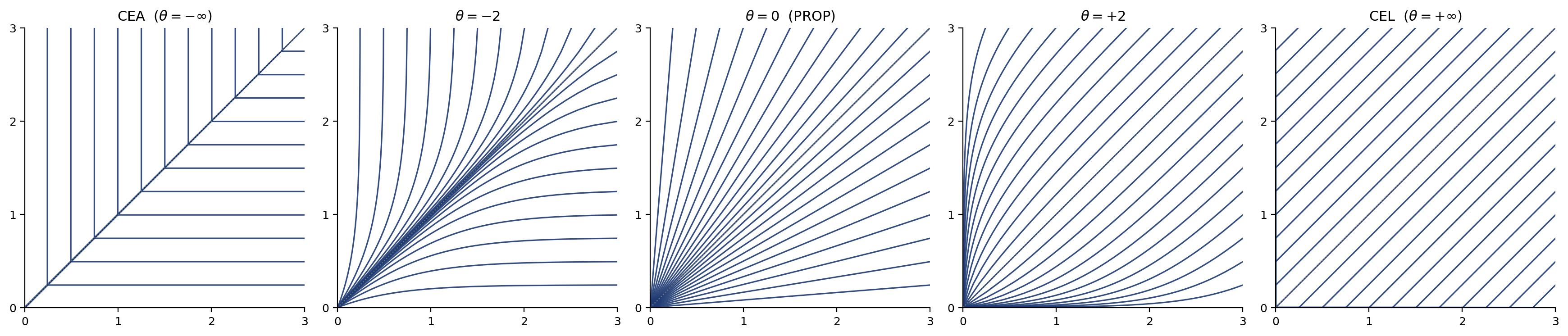}
\caption{Award trees of the symmetric family, ordered by $\theta$. Each displayed
curve is the award path of one claims vector, traced as the endowment runs
from $0$ to $\norm{c}$; the terminal claims vectors shown are on
the boundary of the window. In the full tree, the award path for $c$ is the
unique route from the origin to $c$ along the branches. For finite $\theta$,
distinct branches meet only at the origin; for $\CEA$ and $\CEL$, distinct
award paths share nontrivial segments. At $\theta=0$ the branches are rays.
Along every non-diagonal branch, the ratio of the smaller to the larger
award decreases as one moves outward when $\theta<0$, and increases when
$\theta>0$; equivalently, the branches bend away from the diagonal for
$\theta<0$ and towards it for $\theta>0$.}\label{fig:trees}
\end{figure}

For positive weights $(w_i)_{i\in N^{\ast}}$, the {\bf weighted CEA} and {\bf weighted CEL} rules are given, respectively, by
\begin{equation}\label{eq:weighted}
 r_i(N,c,E)=\min\{c_i,w_i s\},\qquad
 r_i(N,c,E)=(c_i-w_i s)_+,
\end{equation}
where $a_+:=\max\{a,0\}$ and $s\geq0$ is chosen so that the awards sum to $E$.
In either rule the total varies continuously between $0$ and $\norm{c}$,
and monotonicity makes the awards well-defined.
Weighted CEA equalises awards per unit of weight among claimants with unmet
claims. Weighted CEL equalises losses per unit of weight among those with
positive awards. Only the ratios of the weights matter.
\subsection{Chambers' two-claimant rules}\label{sec:chambers}

For a pair of agents $\{i,j\}$, \citet[Section~3.2, Theorem~2]{Chambers2006} characterises the two-claimant rules satisfying CD and CU as {\bf difference rules}. Such a rule is described by a family of continuous, coordinatewise nondecreasing paths $\gamma:\R\to\R^2$. Parametrise each path by $\gamma_i(t)+\gamma_j(t)=t$. One coordinate is convex in $t$ and the other is concave.

For each claims vector $c$, some path has a segment with displacement $c$, so $\gamma(a+\norm{c})-\gamma(a)=c$ for some $a\in\R$. Translating this segment to the origin gives the awards
\begin{equation}\label{eq:chambers}
 r^{ij}(c,E)=\gamma(a+E)-\gamma(a),\qquad 0\le E\le\norm{c}.
\end{equation}
All segments with the same displacement give the same awards after translation. Taking an initial segment gives CD. Translating the remaining segment to the origin gives CU.

\subsection{Priority classes and the characterisation}\label{sec:priority-definitions}
Next we consider rules that partition claimants into groups of different priority so that higher priority  claims have to be fully satisfied before we start to satisfy lower priority claims. On each group a different rule may be applied to allocate within the group of equal priority claims. Formally, 
a {\bf priority composition} is specified by a partition $\mathcal B$ of
$N^\ast$, a strict linear order $\succ$ on its classes, and a division rule
$r^B$ on each class $B$. Write $D\succ B$ if $D$ is paid in full before $B$
receives anything. In a problem $(N,c,E)$ only the (finitely many) classes in
$\mathcal B(N)=\{B\in\mathcal B:N\cap B\ne\varnothing\}$ participate. The
total claim of $B\in\mathcal B(N)$ and the endowment available to it are
\[
 C_B=\sum_{i\in N\cap B}c_i,\qquad
 E_B=\min\Bigl\{C_B,\,\Bigl(E-\sum_{\substack{D\in\mathcal B(N)\\D\succ B}} C_D\Bigr)_+\Bigr\}.
\]
The composed rule assigns $r^B(N\cap B,c_{N\cap B},E_B)$ to $N\cap B$.
The order on $\mathcal B$ need not have a first class, since every problem
involves only finitely many classes.

For example, let $\{1,2\}\succ\{3\}$ and use weighted CEA with weights
$(2,1)$ on $\{1,2\}$. At claims $(6,4,8)$ the awards are $(4,2,0)$ when
$E=6$ and $(6,4,2)$ when $E=12$.

We can now state our main theorem:

\begin{theorem}[Characterisation]\label{thm:main}
A rule satisfies CONS, CD and CU if and only if it is a priority composition
over a linearly ordered partition $\mathcal B$ of $N^\ast$ with the following
rules on its classes.
\begin{enumerate}[label=\textup{(\roman*)}]
\item A class $B$ of at least three claimants uses $r^\theta$ for some profile $\theta\in\R^B$, weighted CEA, or weighted CEL.
\item A class of two claimants uses a difference rule in which neither claimant has absolute priority.
\item A singleton receives its class endowment.
\end{enumerate}
Every such rule satisfies CONT.
\end{theorem}
Section~\ref{sec:priority-partition} establishes the priority partition. Sections~\ref{sec:strict} and \ref{sec:asym-interior} treats the strictly interior case, first for a symmetric rule and then for claimant-specific gain maps. Section~\ref{sec:mainproof} deals with the boundary case. Section~\ref{sec:proof} completes the proof of Theorem~\ref{thm:main}.

\begin{remark}[Continuity in types]\label{rem:types}
The set $N^\ast$ of potential agents may be uncountable. For example, we
can take $N^\ast=\R$. Each claims problem still involves only a finite
set $N\subseteq N^\ast$.

In some applications, agents have types that matter for the rule. We
can take $N^\ast=\mathbb N\times T$, where $(i,a)$ denotes an agent with
index $i$ and type $a\in T$. For example, $T=\R$ gives a one-dimensional
type. The index allows distinct agents $(i,a)$ and $(j,a)$ of the
same type to participate in the same problem.

When $T$ carries a topology, we can ask whether awards vary continuously
with types. Suppose $T$ is connected and the rule satisfies CONS, CD and
CU and gives equal awards to agents with the same type and claim.
Consider a two-agent problem in which both claims and the endowment
are $1$. Suppose the two awards in this problem depend only on the
ordered pair of types $(a,b)$ and are continuous at each $(a,a)$ in the
product topology. Two agents of type $a$ each receive $1/2$ in this problem. They are
therefore tied in priority, since absolute priority would give one of
them the entire endowment. By Theorem~\ref{thm:main}, the priority
classes partition $N^\ast$. Since all agents of a given type belong
to the same class, these classes induce a partition of $T$. By continuity
at $(a,a)$, both awards remain positive for every $b$ in a neighbourhood
of $a$. Agents of types $a$ and $b$ are therefore tied in priority, so the
class containing $a$ contains this neighbourhood. Every class is open
in $T$, so connectedness implies that there is only one priority class.
In the companion paper the type is account equity, which ranges over
$(0,\infty)$
\citep{SchlegelMamageishvili2026}.
\end{remark}

Imposing ETE in addition, we get the following corollary of Theorem~\ref{thm:main}.

\begin{corollary}[Symmetric characterisation]\label{cor:symmetric}
A rule $r$ satisfies \textup{ETE}, \textup{CONS}, \textup{CD} and \textup{CU} if and only if $r=r^\theta$ for a scalar $\theta\in[-\infty,+\infty]$.
\end{corollary}
\begin{proof}
Give any two claimants $i,j$ claims $(2,2)$ and let $E=2$. ETE gives awards $(1,1)$, so neither claimant has absolute priority. All claimants are therefore in one priority class. Since $N^\ast$ has at least three members, part~\textup{(i)} of Theorem~\ref{thm:main} applies. Under weighted CEA or weighted CEL, these awards require $w_i s=w_j s=1$, so $w_i=w_j$. Under a log-exponential rule, \eqref{eq:rtheta} gives
\[
 \lambda=\frac{\varphi_{\theta_k}(1)}{\varphi_{\theta_k}(2)}
        =\frac{1}{1+e^{\theta_k}},\qquad k\in\{i,j\},
\]
so $\theta_i=\theta_j$. Since the pair was arbitrary, all weights or all parameters coincide. Conversely, each symmetric $r^\theta$ satisfies ETE by construction and the other three axioms by Theorem~\ref{thm:main}.
\end{proof}
\begin{remark}[Role of the population bound]\label{rem:popbound}
The assumption of at least three potential agents is important here: With only two potential agents, bilateral consistency imposes no restriction and the two composition axioms allow the larger class of difference rules of \citet{Chambers2006}.
Any bilateral difference rule extends to three agents by giving the third agent lower priority than both original agents.
If all three are to belong to the same priority class, as ETE requires, Theorem~\ref{thm:main} says that the bilateral rule must be log-exponential, weighted CEA, or weighted CEL.
\end{remark}

\begin{remark}[Anonymity from equal treatment]\label{rem:samepair}
Under ETE, CONS and endowment monotonicity, the bilateral rules coincide up to relabelling. To see this, take three distinct agents $i,j,k$ with claims $(a,b,a)$. At every endowment, $i$ and $k$ receive the same award. As the endowment increases, the total awarded to $i,j$ varies continuously from $0$ to $a+b$ (Lemma~\ref{lem:resmon} and \eqref{eq:endowment-lipschitz}). CONS on $\{i,j\}$ and $\{k,j\}$ therefore gives the same bilateral awards at every endowment. Repeating the argument with claims $(a,b,b)$ shows that either agent can be replaced. These replacements connect any two ordered pairs, using a third agent when necessary.\footnote{The proof of Lemma~3 in \citet{ChambersThomson2002} uses five potential agents.} Bilateral consistency then determines the awards on each population (Lemma~\ref{lem:fullcons}). In particular, if every bilateral rule is $r^{\theta^{ij}}$, Proposition~\ref{prop:lorenz} implies that all $\theta^{ij}$ coincide.
\end{remark}

We can also recover the homogeneous characterisation of \citet{Moulin2000} as a corollary of Theorem~\ref{thm:main}. In the symmetric case, dropping homogeneity connects its three isolated rules, $\CEA$, $\PROP$ and $\CEL$, in the continuous family $\{r^\theta\}_{\theta\in[-\infty,+\infty]}$ of Corollary~\ref{cor:symmetric}.

\begin{corollary}[The homogeneous case]\label{cor:moulin}
A rule satisfies homogeneity, \textup{CONS}, \textup{CD} and \textup{CU} if and only if it is a priority composition with the following rules on its classes.
\begin{enumerate}[label=\textup{(\roman*)}]
\item A class of at least three claimants uses $\PROP$, weighted CEA, or weighted CEL.
\item A class of two claimants uses a homogeneous difference rule in which neither claimant has absolute priority.
\item A singleton receives its class endowment.
\end{enumerate}
\end{corollary}
\begin{proof}
A priority composition is homogeneous if and only if the rule on each class is homogeneous. Necessity follows by restricting to problems within a class. For sufficiency, scaling all claims and the endowment by $\alpha>0$ scales each class endowment $E_B$ by $\alpha$, so the awards also scale by $\alpha$.

For a log-exponential profile $\theta$, the identity $\varphi_{\theta_i}(\alpha x)=\alpha\varphi_{\alpha\theta_i}(x)$ gives
\[
 r^\theta(N,\alpha c,\alpha E)=\alpha r^{\alpha\theta}(N,c,E),\qquad \alpha>0.
\]
Homogeneity implies $r^\theta=r^{\alpha\theta}$ for every $\alpha>0$. Fix a problem and let $x=r^\theta(N,c,E)$. By \eqref{eq:rtheta}, the ratios $\varphi_{\alpha\theta_i}(x_i)/\varphi_{\alpha\theta_i}(c_i)$ coincide across all $i$ with $c_i>0$. As $\alpha\downarrow0$, $\varphi_{\alpha\theta_i}(z)\to z$, so the ratios $x_i/c_i$ coincide. Thus the rule is $\PROP$. Weighted CEA and weighted CEL are homogeneous directly from their definitions. The remaining cases follow from Theorem~\ref{thm:main}.
\end{proof}

\subsection{Preliminaries}\label{sec:prelim}
We collect four preliminary facts. First, the axiom system is closed under duality, and the family $r^\theta$ is self-dual under the involution $\theta\leftrightarrow-\theta$. This allows us to use symmetry arguments at several crucial points. Second, we show that CONS, CD and CU imply a Lipschitz bound in claims and the endowment, and hence CONT. Third, we derive the priority partition from the bilateral award paths. Fourth, we introduce the gain and loss maps used throughout the proof, and show that the composition axioms give gain and loss maps a semigroup structure.

\subsubsection{Duality}
For a rule $r$ the {\bf dual rule} is
$r^\ast(N, c,E) := c - r(N, c, \norm{c} - E)$, that is, the rule that for each problem divides the deficit $\norm{c}-E$ as $r$ divides the endowment.

It is straightforward to see that the dual rule of $r^{\theta}$ is $r^{-\theta}$ so that the family is closed under duality  (in particular $\PROP$ is the only self-dual rule in that family and we have the well-known duality between CEL and CEA):
\begin{lemma}[Duality identity]\label{lem:dualident}
For $\theta \in \R$ and $\lambda \in [0,1]$: if $\varphi_\theta(y) = \lambda\,\varphi_\theta(c_0)$ with $0 \le y \le c_0$, then
\begin{equation}\label{eq:dualident}
\varphi_{-\theta}(c_0 - y) = (1-\lambda)\,\varphi_{-\theta}(c_0).
\end{equation}
Consequently $(r^\theta)^\ast = r^{-\theta}$ for every parameter profile $\theta$.
\end{lemma}
\begin{proof}
For $\theta\ne 0$ (the case $\theta = 0$ is trivial): $e^{\theta y} = 1 + \lambda(e^{\theta c_0}-1)$, so
\[ \varphi_{-\theta}(c_0-y) = \frac{1 - e^{-\theta(c_0-y)}}{\theta} = \frac{e^{\theta c_0} - e^{\theta y}}{\theta\, e^{\theta c_0}} = \frac{(1-\lambda)(e^{\theta c_0}-1)}{\theta\, e^{\theta c_0}} = (1-\lambda)\,\varphi_{-\theta}(c_0). \]
For $(r^\theta)^\ast = r^{-\theta}$: let $y := r^\theta(c, \norm{c}-E)$ with parameter $\lambda$. By \eqref{eq:dualident} the losses $c_i - y_i$ satisfy $\varphi_{-\theta_i}(c_i-y_i) = (1-\lambda)\varphi_{-\theta_i}(c_i)$ and sum to $E$; so $c - y = r^{-\theta}(c, E)$, i.e.\ $(r^\theta)^\ast(c,E) = r^{-\theta}(c,E)$.
\end{proof}
For a constant profile, letting $\theta\to\pm\infty$ recovers CEA and CEL.
\begin{proposition}[Endpoints as limits; duality at the endpoints]\label{prop:endpoints}
Pointwise in every problem: $r^\theta \to \CEL$ as $\theta \to +\infty$ and $r^\theta \to \CEA$ as $\theta \to -\infty$. Moreover
\[
\CEA^\ast = \CEL,
\]
so the duality $(r^\theta)^\ast = r^{-\theta}$ extends to $\theta \in [-\infty,+\infty]$.
\end{proposition}

\begin{proof}
\emph{Duality at the endpoints.} If $\sum_i (c_i - s)_+ = \norm{c} - E$, then $c_i - (c_i - s)_+ = \min(c_i, s)$ and $\sum_i \min(c_i,s) = E$; hence $\CEL^\ast(c,E) = c - \CEL(c,\norm{c}-E) = \CEA(c,E)$.

\emph{Limit $\theta \to +\infty$.} Fix $(c,E)$ with $0 < E < \norm{c}$ and $n := |N|$ (the degenerate endowments are trivial). For $\theta > 0$ the solving parameter satisfies $\lambda \in (0,1)$; write $\lambda = e^{-\theta s}$ with $s = s(\theta) > 0$. Solving $e^{\theta y_i} = 1 + \lambda(e^{\theta c_i}-1)$ exactly:
\[
y_i(\theta) = c_i - s + \tfrac{1}{\theta}\log\!\big(1 + (e^{\theta s}-1)\,e^{-\theta c_i}\big).
\]
The logarithm is $\ge 0$, so $y_i \ge c_i - s$; and $y_i \ge 0$ always; so $y_i \ge (c_i-s)_+$. For the upper bound, $1 + (e^{\theta s}-1)e^{-\theta c_i} \le 1 + e^{\theta(s-c_i)} \le 2\max\big(1, e^{\theta(s-c_i)}\big)$, whence
\[ y_i \;\le\; c_i - s + (s-c_i)_+ + \tfrac{\log 2}{\theta} \;=\; (c_i - s)_+ + \tfrac{\log 2}{\theta}. \]
Summing over $i$ and writing $\zeta(s) := \sum_i (c_i - s)_+$:
\[ \zeta\big(s(\theta)\big) \;\le\; E \;\le\; \zeta\big(s(\theta)\big) + \tfrac{n\log 2}{\theta}. \]
The function $\zeta$ is continuous and strictly decreasing on $[0,\max_i c_i]$ from $\norm{c}$ to $0$, so there is a unique $s_E\in(0,\max_i c_i)$ with $\zeta(s_E)=E$, the $\CEL$ parameter. For all sufficiently large $\theta$, the display gives $\zeta(s(\theta))>0$, hence $s(\theta)\in[0,\max_i c_i]$. If $s(\theta)$ did not converge to $s_E$, compactness would give a subsequence converging to some $s\ne s_E$; the display would imply $\zeta(s)=E$, a contradiction. Thus $s(\theta)\to s_E$, and the two-sided bound on $y_i$ gives $y_i(\theta)\to(c_i-s_E)_+=\CEL_i(c,E)$.

\emph{Limit $\theta \to -\infty$.} By Lemma~\ref{lem:dualident} and the case just proved,
\[ r^\theta(c,E) = c - r^{-\theta}\big(c, \norm{c}-E\big) \;\longrightarrow\; c - \CEL\big(c,\norm{c}-E\big) = \CEA(c,E) \]
using $\CEA^\ast = \CEL$, established above.
\end{proof}

As a consequence of our main theorem, the self-duality of our family will imply in particular
the self-duality of the axiom system we choose. However, it will be important to establish this
result independently, as it will allow us to use symmetry arguments in many places of the proof.

\begin{lemma}[Self-duality of the axiom system]\label{lem:duality}
$r^\ast$ is a rule and $(r^\ast)^\ast = r$. Each of \textup{ETE} and \textup{CONS} holds for $r$ iff it holds for $r^\ast$; and $r$ satisfies \textup{CD} iff $r^\ast$ satisfies \textup{CU} \textup{(}and vice versa\textup{)}.
\end{lemma}

\begin{proof}
$r^\ast$ is a rule: $0 \le r_i(c, \norm{c}-E) \le c_i$ gives $0 \le r^\ast_i \le c_i$, and $\sum_i r^\ast_i = \norm{c} - (\norm{c}-E) = E$. Involutivity and ETE are immediate from the formula. For CONS, let $x^\ast = r^\ast(N,c,E)$, so $x^\ast = c - x$ with $x = r(N,c,\norm{c}-E)$. For a pair $\{i,j\}$, CONS for $r$ gives $(x_i,x_j) = r((c_i,c_j), x_i+x_j)$; since $x_i + x_j = (c_i+c_j) - (x^\ast_i + x^\ast_j)$,
\[ (x^\ast_i, x^\ast_j) = (c_i,c_j) - r\big((c_i,c_j), (c_i+c_j) - (x^\ast_i+x^\ast_j)\big) = r^\ast\big((c_i,c_j), x^\ast_i + x^\ast_j\big). \]
For CD/CU, the composition axioms are dual properties \citep{Moulin2000}: whenever a rule satisfies one, its dual satisfies the other. Involutivity turns each implication into an equivalence.
\end{proof}

\subsubsection{Continuity}\label{sec:auto-cont}

A rule is {\bf endowment monotone} if at fixed claims every award is weakly increasing in the endowment. Either composition axiom implies this property (Lemma~\ref{lem:resmon}). Since the awards sum to the endowment, for $E\le E'$ we have
\begin{equation}\label{eq:endowment-lipschitz}
 \norm{r(c,E')-r(c,E)}_1=E'-E.
\end{equation}
In particular, every award path is continuous. Under endowment monotonicity, the bilateral rules determine an allocation uniquely at each endowment (Lemma~\ref{lem:fullcons}). Under CD, two award paths on the same population that meet at an allocation $x$ agree up to that point, since both agree with $r(x,\cdot)$.

\begin{proposition}[Continuity in claims and the endowment]\label{prop:auto-cont}
Let $r$ satisfy CONS, CD and CU. At a fixed endowment, increasing one claim
weakly increases its owner's award and weakly decreases every other award.
For any two problems on the same population,
\begin{equation}\label{eq:continuity}
 \norm{r(c,E)-r(c',E')}_\infty\leq\norm{c-c'}_1+|E-E'|.
\end{equation}
In particular, $r$ satisfies CONT.
\end{proposition}
\begin{proof}
A singleton receives the endowment. Suppose $|N|\geq2$ and increase $c_i$
by $\delta$ to obtain $c'$. At endowment $\norm{c}$, the other claimants together receive at most
$\sum_{j\ne i}c_j$, so $r_i(c',\norm{c})\ge c_i$. Suppose $i$ received less under $c'$ than under $c$ at some
smaller endowment. By continuity in the endowment the two awards to $i$
would agree at a later endowment. The other claimants would then divide
the same total among the same claims. Their awards would therefore agree by
CONS$^\ast$ (by feasibility if only one other claimant remains). By CD the two paths would agree at all earlier endowments,
a contradiction. Thus $i$'s award weakly increases. The others divide a
weakly smaller total, so CONS$^\ast$ and endowment monotonicity imply that
each of their awards weakly decreases.

The deficits at endowment $E$ are $D=\norm{c}-E$ and $D+\delta=\norm{c'}-E$.
Endowment monotonicity of the dual rule and the preceding claim comparison at deficit $D$ give
\[
 c_i+\delta-r_i(c',E)=r_i^\ast(c',D+\delta)
 \ge r_i^\ast(c',D)\ge r_i^\ast(c,D)=c_i-r_i(c,E).
\]
Thus $i$'s loss weakly increases, so the increase in the award is at most $\delta$. The decreases in the other awards sum to that increase, so every award changes by at most $\delta$.

For arbitrary $c,c'$, let $m_k=\max\{c_k,c'_k\}$. At endowment $\min\{E,E'\}$, increase one claim at a time from $c$ to $m$ and from $c'$ to $m$. Each step changes the award vector in $\ell_\infty$ norm by at most the claim increase. These increases sum to $\sum_k(m_k-c_k)+\sum_k(m_k-c'_k)=\norm{c-c'}_1$. The triangle inequality and \eqref{eq:endowment-lipschitz} give \eqref{eq:continuity}.
\end{proof}

\subsubsection{The priority partition}\label{sec:priority-partition}
Given a rule $r$, claimant $i$ has {\bf absolute priority} over $j$,
written $i\succ_r j$, if $i\ne j$ and $r_i^{ij}(c,E)=\min\{c_i,E\}$ for
every bilateral problem. Claimants $i$ and $j$ are {\bf tied in priority},
written $i\sim_r j$, if neither has absolute priority over the other. We let
$i\sim_r i$. Under CONS, $i\succ_r j$ implies that in every population
containing both claimants, $j$ receives a positive award only if $i$'s claim is paid in full.

We use two consequences of the difference-rule representation from Section~\ref{sec:chambers}.

\begin{lemma}[Bilateral award paths]\label{lem:geometry}
Under CD and CU, for two agents $i,j$ and positive claims $c$ each award $0<x_i<c_i$
occurs at a unique endowment, i.e. there is a unique $E>0$ with $r^{i,j}_i(c,E)=x_i$. If there are positive claims $c$ such that $r^{i,j}(c,E=c_i)=(c_i,0)$, then $i\succ_r j$.
\end{lemma}
\begin{proof}
On a segment in \eqref{eq:chambers}, a nondecreasing convex coordinate can be constant on an interval only at
its initial level on the segment. A concave coordinate can be constant
only at its final level. Thus continuity gives every intermediate award,
and each occurs only once.

A path through $(c_i,0)$ pays $i$ alone up to that point and $j$ alone
thereafter. The $i$-coordinate of its generating map has slope $1$ and
then $0$ on intervals of positive length. It must be concave. Since all
its secant slopes are in $[0,1]$, it has slope $1$ everywhere before the
corner and slope $0$ everywhere after it. If the corner is at $t_0$, the segment from $t_0-c'_i$ to $t_0+c'_j$ has displacement $(c'_i,c'_j)$ for any positive claims $c'_i,c'_j$. By \eqref{eq:chambers} this segment gives their award path, so $i\succ_r j$.
\end{proof}
The previous lemma and CONS allows us to show that the tied in priority relation is an equivalence relation and the absolute priority order is a linear order over equivalence classes.
\begin{proposition}[Priority partition]\label{prop:priority}
Under CONS, CD and CU, the relation $\sim_r$ is an equivalence relation. Its classes are linearly
ordered by $\succ_r$, and $r$ is the priority composition of its
restrictions to those classes.
\end{proposition}
\begin{proof}
Suppose $i\sim_r j$. On any bilateral path with positive claims there is
an open interval on which both claims are partially satisfied. Otherwise
one claim would be paid in full before the other receives anything, giving
absolute priority by Lemma~\ref{lem:geometry}. On that interval both
awards increase strictly. Let $\delta>0$ be sufficiently small and choose
endowments $E_0<E_1$ in this interval with awards $x=r(c,E_0)$ and
$y=r(c,E_1)$ such that $y_j-x_j=\delta$. By CU, $y-x=r(c-x,E_1-E_0)$.
For $0\le t\le E_1-E_0$, CU and CD give
\[
 r(c,E_0+t)-x=r(c-x,t)=r(y-x,t).
\]
Thus for claims $y-x$, both awards increase strictly along the entire
award path and $j$'s claim is $\delta$.
If $i\sim_r j$ and $j\sim_r k$ for three distinct claimants, choose such
bilateral claims for both pairs with the same $\delta$. For the resulting claims
$(c_i,c_j=\delta, c_k)$, consider an endowment $E$ where $r_j((c_i,\delta, c_k),E)=\delta/2$.
By CONS, the claims of $i$ and $k$ are both partially satisfied.
Thus $i\sim_r k$.

To see that $\succ_r$ is transitive, take
$i\succ_r j$ and $j\succ_r k$. Choose a problem for $\{i,j,k\}$ at which $j$'s claim is
partially satisfied.  By continuity of award paths this is always possible using positive claims and a suitable endowment. Then $i$ is fully paid and $k$ receives zero, so
Lemma~\ref{lem:geometry} gives $i\succ_r k$.

These two transitivity properties order the equivalence classes. Indeed,
if $i\sim_r j$ and $i\succ_r k$, then $j\sim_r k$ would imply $i\sim_r k$,
and $k\succ_r j$ would imply $i\succ_r j$. Hence $j\succ_r k$. The
same argument applies to the other representative. A positive award in
a later class now requires full payment of every earlier class. This
determines each class endowment $E_B$, and CONS$^\ast$ determines the
allocation within $B$.
\end{proof}

\subsubsection{Gain and loss maps}\label{sec:maps}
A rule $r$ satisfying CONS, CD and CU can be described through gain and loss maps.
A {\bf gain map} for $r$ is a family $T=(T_i)_{i\in N^\ast}$ of functions
$T_i:\R_+\to\R_+$ with $0\le T_i\le\id$ such that, for every population
$N$ and claims vector $c$, $r$ selects the awards $T_i(c_i)$ if the endowment is their sum,
\begin{equation}\label{eq:realisability}
 T(c):=(T_i(c_i))_{i\in N}
 =r\Big(N,c,\sum_{j\in N}T_j(c_j)\Big).
\end{equation}
The same $T_i$ gives claimant $i$'s award in every population containing $i$.
The corresponding {\bf loss map} is $L=\id-T$.
Write $\T$ and $\Lf$ for the sets of gain and loss maps of $r$. Their
projections onto claimant $i$'s coordinate are
$\T_i=\{T_i:T\in\T\}$ and $\Lf_i=\{L_i:L\in\Lf\}.
$
Note that both the zero map and the identity map are gain maps (and loss maps) for each agent. We call a gain (loss) map
{\bf proper} if it is neither the identity nor the zero map. 

For the log-exponential rule $r^\theta$, the gain and loss maps are given by $\T_i=\{ T_{i,\lambda}:\lambda\in[0,1]\}$ and  $\Lf_i=\{ L_{i,\lambda}:\lambda\in[0,1]\}$ with
\begin{equation}\label{eq:lossmaps}
 T_{i,\lambda}(c_0)=\varphi_{\theta_i}^{-1}\big(\lambda\varphi_{\theta_i}(c_0)\big),
 \qquad
 L_{i,\lambda}(c_0)=\varphi_{-\theta_i}^{-1}\big((1-\lambda)\varphi_{-\theta_i}(c_0)\big)
\end{equation}
by \eqref{eq:rtheta} and Lemma~\ref{lem:dualident}. Thus
\begin{equation}\label{eq:finite-identities}
 T_{i,\mu}\circ T_{i,\lambda}=T_{i,\mu\lambda},\qquad
 L_{i,\mu}\circ L_{i,\lambda}=L_{i,\lambda+\mu-\lambda\mu}.
\end{equation}
For weighted CEA,
the gain maps are the {\bf cap maps}
$c_0\mapsto\min\{c_0,w_i s\}$ for $s\geq 0$ and loss maps are the  {\bf shift maps}
$c_0\mapsto(c_0-w_i s)_+$ for $s\geq0$, and vice versa for weighted CEL.

Under ETE the gain maps and loss maps are the same for all agents. Under ETE, we write
$\T_0$ and $\Lf_0$ for the sets of common gain/loss maps of all agents.
Figure~\ref{fig:gain} shows the gain maps for several members of the
symmetric log-exponential family.

\begin{figure}[t]
\centering
\begin{tikzpicture}
\begin{groupplot}[group style={group size=5 by 1, horizontal sep=0.55cm}, gainpanel]

\nextgroupplot[title={$\mathrm{CEA}=r^{-\infty}$}, ytick={1,2}, yticklabels={$1$,$2$},
               ylabel={award $T(c_0)$}, ylabel style={font=\scriptsize}]
  \addplot[gray!55, densely dotted] coordinates {(0,0) (2,2)};        %
  \addplot[s1] coordinates {(0,0) (0.4,0.4) (2,0.4)};
  \addplot[s2] coordinates {(0,0) (0.8,0.8) (2,0.8)};
  \addplot[s3] coordinates {(0,0) (1.2,1.2) (2,1.2)};
  \addplot[s4] coordinates {(0,0) (1.6,1.6) (2,1.6)};

\nextgroupplot[title={$r^{-1}$}]
  \addplot[gray!55, densely dotted] coordinates {(0,0) (2,2)};
  \addplot[s1]{-ln(1 - 0.381282*(1-exp(-x)))};
  \addplot[s2]{-ln(1 - 0.636860*(1-exp(-x)))};
  \addplot[s3]{-ln(1 - 0.808187*(1-exp(-x)))};
  \addplot[s4]{-ln(1 - 0.923022*(1-exp(-x)))};

\nextgroupplot[title={$\mathrm{PROP}=r^{0}$}]
  \addplot[gray!55, densely dotted] coordinates {(0,0) (2,2)};
  \addplot[s1]{0.2*x};
  \addplot[s2]{0.4*x};
  \addplot[s3]{0.6*x};
  \addplot[s4]{0.8*x};

\nextgroupplot[title={$r^{+1}$}]
  \addplot[gray!55, densely dotted] coordinates {(0,0) (2,2)};
  \addplot[s1]{ln(1 + 0.076979*(exp(x)-1))};
  \addplot[s2]{ln(1 + 0.191818*(exp(x)-1))};
  \addplot[s3]{ln(1 + 0.363137*(exp(x)-1))};
  \addplot[s4]{ln(1 + 0.618719*(exp(x)-1))};

\nextgroupplot[title={$\mathrm{CEL}=r^{+\infty}$}]
  \addplot[gray!55, densely dotted] coordinates {(0,0) (2,2)};
  \addplot[s1] coordinates {(0,0) (1.6,0) (2,0.4)};
  \addplot[s2] coordinates {(0,0) (1.2,0) (2,0.8)};
  \addplot[s3] coordinates {(0,0) (0.8,0) (2,1.2)};
  \addplot[s4] coordinates {(0,0) (0.4,0) (2,1.6)};

\end{groupplot}

\node[font=\scriptsize, anchor=north]
  at ($(group c3r1.south)+(0,-0.10cm)$) {claim $c_0$};
\end{tikzpicture}
\caption{Gain maps, ordered by $\theta$. For finite $\theta$ the utility is $u_\theta=\log\varphi_\theta$. Each curve gives the map with award $\sigma$ at the reference claim $c_0=2$, for $\sigma\in\{0.4,0.8,1.2,1.6\}$. The dotted line is the identity $T(c_0)=c_0$.}\label{fig:gain}
\end{figure}

To obtain a gain map $T$ from $r$, start with an allocation in which at
least three claims are partially satisfied. Give agent $i$ claim $c_0$ and
keep the original claim and award of another agent $j$ whose claim
is partially satisfied. By CONS their
awards lie on the same bilateral path, where $j$'s award determines $i$'s
uniquely (Lemma~\ref{lem:geometry}). Lemma~\ref{lem:realisability} in Appendix~\ref{app:gain-maps}
shows that every choice of $j$ gives the same value $T_i(c_0)$.

\begin{lemma}[Parametrisation of the gain maps]\label{lem:common}
Let $r$ be strictly interior and satisfy CONS, CD and CU. Let $j\in N^\ast$ and  $\bar c_0>0$. For each $a\in[0,\bar c_0]$, the award $a$ to claimant $j$ on claim $\bar c_0$ uniquely determines the gain map $T_a$. For each $i\in N^\ast$ and $c_0>0$, the award $T_{i,a}(c_0)$ increases strictly and continuously from $0$ to $c_0$ as $a$ runs from $0$ to $\bar c_0$. Thus $T_a$ is proper for every agent exactly when $0<a<\bar c_0$.

Under ETE the coordinates coincide, so $T_a$ is the unique common gain map in $\T_0$ with $T_a(\bar c_0)=a$.
\end{lemma}
\begin{proof}
Give $j$ and two other claimants claim $\bar c_0$ each. For $0<a<\bar c_0$, choose an endowment at which $j$ receives $a$. Both other claims are partially satisfied, and $a$ determines their awards on their respective bilateral paths with $j$ (CONS and Lemma~\ref{lem:geometry}). By Lemma~\ref{lem:realisability}, this allocation determines a unique gain map $T_a$. At $a=0$ and $a=\bar c_0$, strict interiority on pairs forces the zero and identity maps, respectively.

For $i\ne j$, the pairs $(a,T_{i,a}(c_0))$ trace the bilateral award path for claims $(\bar c_0,c_0)$. Both awards increase strictly and continuously from zero to their claims by strict interiority and Lemma~\ref{lem:geometry}. For $i=j$, use any other claimant $k$. As $a$ runs from zero to $\bar c_0$, so does $T_{k,a}(\bar c_0)$, and the bilateral path for $k,j$ with claims $(\bar c_0,c_0)$ gives the result. Under ETE, equal claims on any pair give equal awards, so all coordinates of $T_a$ coincide.
\end{proof}

The composition axioms compare applying the rule twice with applying it once. Under CD the second application divides the awards of the first, so the gain maps compose and $\T_i$ is a semi-group. Under CU the same argument applies to losses and $\Lf_i$ is a semi-group.

\begin{lemma}[Composition]\label{lem:criterion}
\begin{enumerate}[label=\textup{(\alph*)}]
\item If $r$ satisfies CD, then \eqref{eq:realisability} is preserved under componentwise composition of maps on any fixed set of claimants. In particular, $\T_i$ is a semigroup. Conversely, if a collection of vectors satisfying \eqref{eq:realisability} is closed under componentwise composition and generates every allocation of $r$, then $r$ satisfies CD.
\item If $r$ satisfies CU, the same closure holds for \eqref{eq:realisability} with the dual rule in place of $r$, and $\Lf_i$ is a semigroup. The converse in (a) applied to the dual rule gives CU for $r$.
\end{enumerate}
\end{lemma}
\begin{proof}
If $T$ and $S$ satisfy \eqref{eq:realisability}, the rule selects $T(c)$ from claims $c$ and $S(T(c))$ from claims $T(c)$. By CD it selects $S(T(c))$ from $c$ at the same total. Thus $S\circ T$ satisfies \eqref{eq:realisability}. By the definition of the dual rule, $\id-T$ satisfies \eqref{eq:realisability} for the dual whenever $T$ satisfies it for $r$. This gives the closure for loss maps under CU.

For the converse, let $E'\le E$, choose $T$ generating $r(c,E)$ and $S$ generating $r(T(c),E')$. Since $S\circ T$ satisfies \eqref{eq:realisability}, it gives the same awards from claims $c$ at endowment $E'$. This is CD. The same argument gives CD for the dual rule and hence CU for $r$ under the corresponding hypothesis for loss maps.
\end{proof}

\section{Proof}

By Proposition~\ref{prop:priority}, it suffices to characterise the rule within each priority class of at least three claimants. We distinguish the boundary case and the strictly interior case.

In the strictly interior case, we begin with symmetric rules (Section~\ref{sec:strict}). For such a rule, composition down gives an equal-sacrifice representation in awards by the classical result of \cite{Young1988} (Proposition~\ref{prop:flow}). By duality, composition up gives an equal-sacrifice representation in losses (Corollary~\ref{cor:dualflow}). The two representations are connected through the "bridge" equation~\eqref{eq:starstar}. We then solve this equation. The utility functions in
the dual equal-sacrifice representations are strictly concave (Lemma~\ref{lem:concave}) and continuously differentiable (Lemma~\ref{lem:C1}).
Equation~\eqref{eq:starstar} in differentiated form becomes the marginal utility
equation~\eqref{eq:G} whose solutions are of the log-exponential form, by Propositions~\ref{prop:marginal} and~\ref{thm:GsolveC0}. In the asymmetric case (Section~\ref{sec:asym-interior}), we first index the claimant-specific gain maps by a common parameter. We then apply the symmetric argument to each claimant's gain maps and show that the same coefficient $\lambda$ in Equation~\eqref{eq:rtheta} works for all claimants.

In the boundary case (Section~\ref{sec:boundary}), we start from a gain map through an allocation in which every claim is partially satisfied. Its iterates converge to either the zero map or a cap map that is proper for all agents. A proper cap map gives weighted CEA. By duality, a proper cap map for the dual rule gives weighted CEL. If the cap maps for the rule and its dual are both zero, the rule is strictly interior (Proposition~\ref{prop:trichotomy}).

In Section~\ref{sec:proof}, we verify that the priority compositions in Theorem~\ref{thm:main} satisfy the three axioms. Section~\ref{sec:independence} establishes their independence.

\subsection{The symmetric strictly interior case}\label{sec:strict}
Throughout this subsection, let $r$ be a strictly interior rule on a finite set of at least three claimants satisfying ETE, CONS, CD and CU. By Lemma~\ref{lem:duality}, the dual rule $r^\ast$ satisfies the same axioms. It is also strictly interior, since an award lies strictly between zero and its claim exactly when the corresponding loss does. Its gain maps are the loss maps of $r$.

A main ingredient in our proof is Young's equal-sacrifice representation, read as an additive representation of the gain maps. As we have observed, CD turns the common gain maps into a semigroup. By Lemma~\ref{lem:common}, we can parametrise the elements of $\T_0$ by fixing a positive claim and taking as parameter the award that each gain map assigns to it. Acz\'el's associativity theorem gives an additive representation if the operation is continuous and strictly increasing in each argument. Strict monotonicity of the operation is the demanding requirement in our context. \citet{Young1988} guarantees strict monotonicity by requiring the axioms of strict resource monotonicity and strict order preservation in awards. Here strict interiority together with CD gives strict resource monotonicity, as shown in the proof of Lemma~\ref{lem:common}. The contraction lemma below gives strict order preservation in awards. We use the contraction inequality later to prove strict concavity of the utilities.

\begin{lemma}[Contraction Lemma]\label{lem:contr}
Under the standing assumptions, every proper $T\in\T_0$ and every proper $L\in\Lf_0$ is strictly increasing on $(0,\infty)$. Consequently, for every proper $T\in\T_0$ and all $0<c_0<c_0'$,
\begin{equation}\label{eq:twosided}
 0 \;<\; T(c_0')-T(c_0) \;<\; c_0'-c_0\,.
\end{equation}
\end{lemma}
\begin{proof}

Let $T\in\T_0$ be proper. If $T(c_0)=T(c_0')=:v$ for $0<c_0<c_0'$, Lemma~\ref{lem:common} gives a proper $S\in\T_0$ with $S(c_0')=c_0$. By Lemma~\ref{lem:criterion}(a), $T\circ S\in\T_0$, and
\[
 (T\circ S)(c_0')=T(c_0)=T(c_0').
\]
At the reference claim $\bar c_0=c_0'$ this common award identifies the gain map by Lemma~\ref{lem:common}. Thus $T\circ S=T$, hence $T\circ S^{\circ n}=T$ for every $n$. Strict interiority implies $0<S(c_0)<c_0$ for $c_0>0$, so $S^{\circ n}(c_0')\downarrow\ell$ for some $\ell\ge0$. Continuity gives $S(\ell)=\ell$, and strict interiority forces $\ell=0$. Therefore
\[
 v=T(c_0')=T(S^{\circ n}(c_0'))\longrightarrow T(0)=0,
\]
a contradiction. Thus $T$ is injective. A continuous injective map on an interval is strictly monotone, and $T(0)=0<T(c_0)$ for $c_0>0$ rules out the decreasing case; hence $T$ is strictly increasing. The dual $r^\ast$ satisfies the same assumptions and has gain maps $\Lf_0$. The statement just proved, applied to $r^\ast$, therefore makes every proper $L\in\Lf_0$ strictly increasing. Since $T$ is proper if and only if $\id-T$ is proper, \eqref{eq:twosided} is the conjunction of strict increase of $T$ and of $\id-T$ at $c_0<c_0'$.
\end{proof}

With strict order preservation established, we can apply Young's theorem to represent every proper gain map by a common utility loss.

\begin{proposition}[Young--Acz\'el equal-sacrifice representation]\label{prop:flow}
Under the standing assumptions, the set $\T_0$ can be parametrised as
\[
 \T_0=\{\,G_t:t\in[0,\infty)\,\}\cup\{G_\infty\},
 \qquad
 G_t(c_0):=U^{-1}\big(U(c_0)-t\big)\ (c_0>0),\quad G_t(0):=0,\quad G_\infty:=0,
\]
for a continuous strictly increasing $U:(0,\infty)\to\R$ with $U(0^+)=-\infty$, so that $\T_0$ is a one-parameter semigroup:
\[
 U\big(G_t(c_0)\big)=U(c_0)-t \quad(t<\infty),\qquad
 G_t\circ G_s=G_{t+s}\quad(t,s\in[0,\infty]),\qquad G_0=\id.
\]
For each $c_0>0$, $t\mapsto G_t(c_0)$ is a continuous strictly decreasing bijection of $(0,\infty)$ onto $(0,c_0)$.
\end{proposition}

\begin{proof}
Young's theorem is formulated for arbitrary finite populations, whereas our rule may have only finitely many potential agents. To extend $r$ to Young's domain, fix $\bar c_0>0$ and index the common gain maps by $a=T_a(\bar c_0)\in[0,\bar c_0]$ as in Lemma~\ref{lem:common}. On each finite population, let the awards be $T_a(c_i)$, choosing $a$ so that $\sum_i T_a(c_i)=E$. The sum increases continuously from zero to $\norm{c}$ and is strictly increasing if any claim is positive. This defines a rule $\hat r$ agreeing with $r$ on its domain by \eqref{eq:realisability}. The maps are $1$-Lipschitz in the claim and continuous in $a$, so they are jointly continuous. Compactness of $[0,\bar c_0]$ and uniqueness of the awards give continuity of the rule.

The assumptions on $r$ hold on its original domain. On arbitrary finite populations, $\hat r$ satisfies ETE because it uses the same map for every claimant. It satisfies CONS$^\ast$ because the same $a$ balances the budget of any subgroup at its total award. Closure of the gain maps under composition gives CD for $\hat r$ by Lemma~\ref{lem:criterion}(a).

For $c_0>0$, Lemma~\ref{lem:common} gives $0<T_a(c_0)<c_0$ for $0<a<\bar c_0$, and the award $T_a(c_0)$ increases strictly with $a$, hence with the endowment. For a positive endowment, the gain map is the identity or is proper and strictly increasing in the claim by Lemma~\ref{lem:contr}. On Young's domain of strictly positive incomes, define the tax $F_i(c,\tau)=c_i-\hat r_i(c,\norm{c}-\tau)$. Each tax increases strictly with $\tau$, and for $\tau<\norm{c}$ the after-tax incomes preserve the strict ordering of pre-tax incomes. By CD, for $0\le\tau\le\tau'\le\norm{c}$ it satisfies
\[
 F(c,\tau')-F(c,\tau)=F\big(c-F(c,\tau),\tau'-\tau\big).
\]

Young's equal-sacrifice theorem \citep[Theorem~1]{Young1988} yields a continuous strictly increasing $U:(0,\infty)\to\R$ such that every positive-claims problem with an interior endowment equalises
$
 U(c_i)-U\big(\hat r_i(c,E)\big)
$
across claimants. Since $\hat r$ agrees with $r$ on its original domain, the representation applies there as well.

Fix a proper $T\in\T_0$. For arbitrary $a,b>0$, realise $T$ on the two-claimant problem with claims $(a,b)$ and endowment $T(a)+T(b)$. Equal sacrifice gives
\[
 U(a)-U(T(a))=U(b)-U(T(b)).
\]
Thus there is a number $t(T)>0$, independent of the claim, such that
\begin{equation}\label{eq:youngabel}
 U(T(c_0))=U(c_0)-t(T) \qquad(c_0>0).
\end{equation}

The boundary behaviour follows directly in our setting. For any proper $T\in\T_0$ and $c_0>0$, the sequence $T^{\circ n}(c_0)$ decreases to a limit $\ell\ge0$. Continuity gives $T(\ell)=\ell$, and strict interiority forces $\ell=0$. Iterating \eqref{eq:youngabel} gives
\[
U\bigl(T^{\circ n}(c_0)\bigr)=U(c_0)-n\,t(T)\longrightarrow-\infty.
\]
Since $U$ is increasing, its extended right limit at zero exists and equals $-\infty$. Thus $U((0,c_0])=(-\infty,U(c_0)]$, so $U^{-1}(U(c_0)-t)$ is defined for every $c_0>0$ and finite $t\ge0$.

By Lemma~\ref{lem:common}, the values $T(\bar c_0)$ of the proper gain maps range over $(0,\bar c_0)$. Hence, using \eqref{eq:youngabel},
\[
 \{t(T):T\in\T_0\text{ proper}\}
 =\{U(\bar c_0)-U(\sigma):0<\sigma<\bar c_0\}
 =(0,\infty).
\]
For each $t>0$, equation \eqref{eq:youngabel} therefore identifies a unique proper map
\[
 G_t(c_0):=U^{-1}\big(U(c_0)-t\big).
\]
Together with $G_0:=\id$ and $G_\infty:=0$, these are exactly the members of $\T_0$. For finite $s,t$,
\[
 U\big(G_t(G_s(c_0))\big)=U(c_0)-s-t,
\]
so injectivity of $U$ gives $G_t\circ G_s=G_{t+s}$; the endpoint conventions are immediate. Finally, for fixed $c_0>0$, continuity and strict increase of $U^{-1}$ show that $t\mapsto G_t(c_0)$ is continuous and strictly decreasing, with limits $c_0$ as $t\downarrow0$ and $0$ as $t\to\infty$. Its range is therefore $(0,c_0)$.
\end{proof}

\begin{remark}[Normalisation of $U$]
For the $r^{\theta}$ family, $U$ is a positive affine transformation of $\log\varphi_\theta$. In the normalisation $U=\log\varphi_\theta$ one has $G_t(c_0)=\varphi_\theta^{-1}\big(e^{-t}\varphi_\theta(c_0)\big)$, and the gain map at budget parameter $\lambda$ is $G_{-\log\lambda}$. Replacing $U$ by $aU+b$ with $a>0$ rescales the sacrifice parameter $t$ by $a$.
\end{remark}
The second composition axiom now adds the decisive structure: by duality, the loss maps have an analogous representation, and combining the two utilities gives a functional equation with considerable force:
\begin{corollary}[Dual representation; bridge equation]\label{cor:dualflow}
Under the standing assumptions, there is a continuous strictly increasing $V : (0,\infty)\to\R$, $V(0^+)=-\infty$, such that the proper members of $\Lf_0$ are $\{H_u := V^{-1}(V(\cdot)-u) : u \in (0,\infty)\}$. Each proper loss $L_t := \id - G_t$ equals $H_{\beta(t)}$ for a strictly decreasing bijection $\beta : (0,\infty)\to(0,\infty)$, so that
\begin{equation}\label{eq:starstar}
V\big(c_0 - G_t(c_0)\big) = V(c_0) - \beta(t) \qquad \text{for all } c_0>0,\ t>0. \tag{$\star\star$}
\end{equation}
\end{corollary}

\begin{proof}
Apply Proposition~\ref{prop:flow} to the dual rule $r^\ast$, whose gain maps are the loss maps of $r$. The resulting maps $\{H_u\}_{u>0}$ and the maps $\{L_t:=\id-G_t\}_{t>0}$ therefore parametrise the same proper loss maps. Hence there is a unique bijection $\beta:(0,\infty)\to(0,\infty)$ such that $L_t=H_{\beta(t)}$. As $t$ increases, $L_t$ increases pointwise, whereas $H_u$ decreases with $u$, so $\beta$ is strictly decreasing. Writing $L_t=H_{\beta(t)}$ in the $V$-coordinate gives \eqref{eq:starstar}.
\end{proof}
\subsubsection{The marginal-utility equation}\label{sec:smooth}

If $U$ and $V$ are $C^1$ with strictly positive derivatives, then differentiating equation \eqref{eq:starstar} in $c_0$ at fixed $t$ gives
\[
V'\big(c_0-G_t(c_0)\big)\big(1-\partial_{c_0} G_t(c_0)\big)=V'(c_0),
\]
while differentiating $U(G_t(c_0)) = U(c_0)-t$ \textup{(}Proposition~\ref{prop:flow}\textup{)} yields $\partial_{c_0} G_t(c_0)=U'(c_0)/U'(G_t(c_0))$. Let $y:=G_t(c_0)$. As $t$ sweeps $(0,\infty)$, $y$ sweeps $(0,c_0)$. Eliminating $\partial_{c_0} G_t$ gives
\begin{equation}\label{eq:G}
\frac{U'(c_0)}{U'(y)} + \frac{V'(c_0)}{V'(c_0-y)} = 1, \qquad 0 < y < c_0. \tag{G}
\end{equation}

One checks that the pair $U'=1/(C\varphi_{-\theta})$, $V'=1/(\widetilde C\varphi_{\theta})$ \textup{(}for constants $C, \widetilde C > 0$\textup{)} solves \eqref{eq:G}; in fact the verification is nothing but the duality identity: the two ratios in \eqref{eq:G} become $\varphi_{-\theta}(y)/\varphi_{-\theta}(c_0)$ and $\varphi_\theta(c_0-y)/\varphi_\theta(c_0)$, and if $\lambda$ denotes the first, Lemma~\ref{lem:dualident}, applied with the sign of $\theta$ reversed, evaluates the second as $1-\lambda$. By the logarithmic-derivative identity
\begin{equation}\label{eq:logderiv}
\frac{\varphi_\theta'(x)}{\varphi_\theta(x)} = \frac{1}{\varphi_{-\theta}(x)},
\end{equation}
this choice gives $U' = \tfrac1C\,\varphi_\theta'/\varphi_\theta$, i.e.\ $U = \tfrac1C\log\varphi_\theta + \text{const}$, so that $U(G_t(c_0)) = U(c_0)-t$ becomes $\varphi_\theta(G_t(c_0)) = e^{-Ct}\varphi_\theta(c_0)$. Thus, with $\lambda=e^{-Ct}$, these maps are exactly the gain maps of $r^\theta$; the constant $C$ merely rescales the parameter $t$.

\begin{proposition}[Marginal-utility equation and candidate solutions]\label{prop:marginal}
Under the standing assumptions, let $U$ and $V$ be as obtained in Proposition~\ref{prop:flow} and Corollary~\ref{cor:dualflow}. If $U$ and $V$ are $C^1$ on $(0,\infty)$ with strictly positive derivatives, then \eqref{eq:G} holds. Moreover, the pairs $U'=1/(C\varphi_{-\theta})$, $V'=1/(\widetilde C\varphi_{\theta})$ \textup{(}$C, \widetilde C > 0$, $\theta\in\R$\textup{)} solve \eqref{eq:G}, and whenever $U = \tfrac1C\log\varphi_\theta$ up to an additive constant, the proper members of $\T_0$ are exactly the proper gain maps of $r^\theta$.
\end{proposition}

\begin{proof}
The calculation above; for the last statement, $\varphi_\theta(G_t(c_0)) = e^{-Ct}\varphi_\theta(c_0)$ identifies $G_t = \varphi_\theta^{-1}\big(\lambda\,\varphi_\theta(\cdot)\big)$ with $\lambda = e^{-Ct}$, and $t\in(0,\infty)$ corresponds to $\lambda\in(0,1)$.
\end{proof}

\subsubsection{Regularity of the utilities}\label{subsec:autoreg}
To justify the marginal-utility equation of the previous subsection, we now establish regularity of $U$ and $V$. First,
the contraction lemma, written in the equal-sacrifice representation, implies strict concavity of both utilities. Consequently, they are locally Lipschitz and differentiable almost everywhere.
\begin{lemma}[Concavity of the utilities]\label{lem:concave}
$U$ and $V$ are strictly concave on $(0,\infty)$.
\end{lemma}

\begin{proof}
Fix $0<x<z$ and let $m:=(x+z)/2$ and $t:=U(z)-U(m)>0$. Then $G_t(z)=m$. If
\[
U(m)\le \frac{U(x)+U(z)}2,
\]
then
\[
U\bigl(G_t(m)\bigr)=U(m)-t=2U(m)-U(z)\le U(x),
\]
so $G_t(m)\le x$. It follows that
\[
G_t(z)-G_t(m)\ge m-x=z-m,
\]
contradicting the strict contraction inequality in Lemma~\ref{lem:contr}. Thus $U$ is strictly midpoint-concave. Continuity gives concavity, and strict midpoint concavity rules out affine subintervals; hence $U$ is strictly concave.

Apply the same argument to the dual rule $r^\ast$ and its gains $\{H_u\}$; Lemma~\ref{lem:contr} supplies the same strict contraction inequality. Thus $V$ is strictly concave as well.
\end{proof}

A finite concave function $U$ on an open interval is locally Lipschitz and has finite one-sided derivatives everywhere, with $U'_-$ and $U'_+$ nonincreasing. See \citet[Chapter~1]{RobertsVarberg1973} for these classical facts. If $U'_+(y_1) \le 0$ at some $y_1$, then $U'_+ \le 0$ on $[y_1,\infty)$ and $U$ is nonincreasing there, contradicting strict increase; hence $U'_- \ge U'_+ > 0$ everywhere, and likewise for $V$. The next lemma uses the bridge equation to rule out corners in both utilities. This supplies the regularity needed to differentiate the bridge equation and obtain \eqref{eq:G}.

\begin{lemma}[$C^1$ regularity]\label{lem:C1}
$U$ and $V$ are $C^1$ on $(0,\infty)$ with continuous, strictly positive derivatives.
\end{lemma}

\begin{proof}
We claim that $U$ has no corner. Suppose for contradiction that $U$ has a corner at $y_1$, i.e.\ $U'_-(y_1) > U'_+(y_1)$. Since $U$ is continuous and strictly increasing, $c_0 \mapsto \tau(c_0) := U(c_0) - U(y_1)$ is a homeomorphism from $(y_1, \infty)$ onto $(0, L)$, where $L := \sup U - U(y_1) \in (0,\infty]$. Since the monotone function $\beta$ is differentiable almost everywhere, choose $t_0\in(0,L)$ at which its derivative is finite, and let $c_0=\tau^{-1}(t_0)$.

By Proposition~\ref{prop:flow}, the substitution $y=G_t(c_0)$ gives $t=U(c_0)-U(y)$ for every $0<y<c_0$. The bridge equation \eqref{eq:starstar} becomes
\begin{equation*}\tag{$\star\star'$}
V(c_0 - y) = V(c_0) - \beta\bigl(U(c_0) - U(y)\bigr), \qquad 0 < y < c_0.
\end{equation*}
Take one-sided $y$-derivatives at $y = y_1$ with $c_0$ fixed. As $y$ increases through $y_1$, the argument $c_0 - y$ decreases through $c_0 - y_1$, so the left-hand side has one-sided derivatives
\[
\frac{d^+}{dy}\Big|_{y_1} V(c_0 - y) = -V'_-(c_0 - y_1), \qquad
\frac{d^-}{dy}\Big|_{y_1} V(c_0 - y) = -V'_+(c_0 - y_1).
\]
Since $\beta$ is differentiable at $\tau(c_0) = U(c_0) - U(y_1)$, the chain rule for one-sided derivatives gives, for the composite term,
\[
\frac{d^+}{dy}\Big|_{y_1} \beta\bigl(U(c_0) - U(y)\bigr) = -\beta'(\tau(c_0))\, U'_+(y_1), \qquad
\frac{d^-}{dy}\Big|_{y_1} \beta\bigl(U(c_0) - U(y)\bigr) = -\beta'(\tau(c_0))\, U'_-(y_1).
\]
Matching one-sided derivatives of both sides of $(\star\star')$ and defining $\alpha:=-\beta'(\tau(c_0))\ge0$ we obtain
\[
V'_-(c_0-y_1)=\alpha U'_+(y_1),
\qquad
V'_+(c_0-y_1)=\alpha U'_-(y_1).
\]
The first equality and strict positivity of $V'_-$ and $U'_+$ imply $\alpha>0$. The assumed corner of $U$ then gives $V'_+(c_0-y_1)>V'_-(c_0-y_1)$, contradicting concavity of $V$. Thus $U$ has no corner.

By self-duality, the dual bridge equation is the original bridge equation {$(\star\star')$} with $(U, V, \beta)$ replaced by $(V, U, \beta^{-1})$. Since $\beta^{-1}$ is a strictly decreasing bijection, it is differentiable a.e., and the argument above, applied with $(V, U, \beta^{-1})$ in place of $(U, V, \beta)$, shows that $V$ has no corner.

A concave function on an open interval that is differentiable everywhere is $C^1$: its derivative is monotone (nonincreasing for a concave function) and has no jumps, hence is continuous; see \citet[Chapter~1]{RobertsVarberg1973}. The derivatives are strictly positive and continuous throughout.
\end{proof}

\subsubsection{Uniqueness}\label{sec:uniqueness}

Lemmas~\ref{lem:concave} and~\ref{lem:C1} establish the hypothesis of Proposition~\ref{prop:marginal}, so equation \eqref{eq:G} holds. It remains to show that the pairs listed there are its \emph{only} solutions for these utilities.

\begin{proposition}[Uniqueness]\label{thm:GsolveC0}
Let $U,V:(0,\infty)\to\mathbb R$ be $C^1$, strictly increasing and strictly concave, with strictly positive derivatives, and suppose that \eqref{eq:G} holds. Then there exist $\theta\in\mathbb R$, $C,\widetilde C>0$ and $a,b\in\mathbb R$ such that
\[
U(x)=\frac1C\log\varphi_\theta(x)+a,
\qquad
V(x)=\frac1{\widetilde C}\log\varphi_{-\theta}(x)+b.
\]
\end{proposition}

\begin{proof}
Let
\[
g:=\frac1{U'},\qquad h:=\frac1{V'},\qquad q:=\frac gh.
\]
The functions $g$ and $h$ are positive, continuous and strictly increasing. Writing $c_0=x+y$ in \eqref{eq:G} and rearranging gives
\begin{equation}\label{eq:h-cocycle}
h(x+y)=h(y)+\frac{q(x)}{q(x+y)}h(x),
\qquad x,y>0. \tag{H}
\end{equation}
We first show that, after normalisation, $q$ satisfies the exponential Cauchy equation. This determines $q$ up to a multiplicative constant. Equation~\eqref{eq:h-cocycle} then gives $h=\widetilde C\varphi_\theta$ and $g=C\varphi_{-\theta}$, and integration gives the utilities.

Evaluate $h(x+y+z)$ from \eqref{eq:h-cocycle} first as $(x+z)+y$ and then as $x+(z+y)$. Cancelling the common terms yields
\begin{align}
\frac{q(x+z)}{q(x+y+z)}=\frac{q(z)}{q(y+z)},
\qquad x,y,z>0.\label{ratio}
\end{align}
Fix $a>0$. Applying the ratio equation with $x=y=a-t$ and $z=t$ for $0<t<a$ and rearranging gives $q(t)=q(a)^2/q(2a-t).$
Hence the limit
\[
q(0^+):=\lim_{t\downarrow0}q(t)=\frac{q(a)^2}{q(2a)}\in(0,\infty)
\]
exists. Put $\bar q(x)=q(x)/q(0^+)$ and $\bar q(0)=1$. Letting $z\downarrow0$ in \eqref{ratio} gives the exponential Cauchy equation \[
\bar q(x+y)=\bar q(x)\bar q(y).
\]
Under continuity it has exactly the solutions $\bar q(x)=e^{-\theta x}$ for some $\theta\in\mathbb R$; see \citet[Section~2.1]{Aczel1966}. Thus $q(x)=q(0^+)e^{-\theta x}$.

Substitution in \eqref{eq:h-cocycle} gives
\[
h(x+y)=h(y)+e^{\theta y}h(x).
\]
Interchanging $x$ and $y$ therefore gives
\[
(e^{\theta y}-1)h(x)=(e^{\theta x}-1)h(y).
\]
If $\theta\ne0$, this implies $h=\widetilde C\varphi_\theta$ for some $\widetilde C>0$. If $\theta=0$, equation~\eqref{eq:h-cocycle} is additive Cauchy, so continuity gives the same conclusion with $\varphi_0(x)=x$. Since $g=qh$, in either case $g=C\varphi_{-\theta}$ for some $C>0$.

Finally, the logarithmic-derivative identity \eqref{eq:logderiv} gives
\[
U'=\frac1{C\varphi_{-\theta}}
=\frac1C\frac{\varphi_\theta'}{\varphi_\theta},
\qquad
V'=\frac1{\widetilde C\varphi_\theta}
=\frac1{\widetilde C}\frac{\varphi_{-\theta}'}{\varphi_{-\theta}}.
\]
Integration yields the stated forms of $U$ and $V$.
\end{proof}

\begin{proposition}[The strictly interior case]\label{prop:interior}
Under the standing assumptions, $r=r^\theta$ for some $\theta\in\R$.
\end{proposition}

\begin{proof}
By Lemma~\ref{lem:C1} and Propositions~\ref{prop:marginal} and~\ref{thm:GsolveC0}, the proper members of $\T_0$ are exactly the proper gain maps of $r^\theta$ for some $\theta\in\R$. The zero and identity maps are also common to both rules.

For any problem with at least two claimants, choose a gain map $T\in\T_0$ generating its awards under $r$. This map is also a gain map of $r^\theta$ and satisfies the same budget equation, so it generates the same awards under $r^\theta$. Singleton problems are fixed by budget balance. Hence $r=r^\theta$. Uniqueness of $\theta$ follows from Proposition~\ref{prop:lorenz}.
\end{proof}

\begin{remark}[Supermodularity and copulas]\label{rem:copula}
The semigroup representation and the composition axioms imply that gains are supermodular. At fixed claims, an increase in the endowment corresponds to a decrease in sacrifice from $t_2$ to $t_1<t_2$. The additional award on claim $c_0$ equals the reduction in its loss. It is strictly increasing in $c_0$. Indeed, the semigroup identity of Proposition~\ref{prop:flow} gives
\[
 G_{t_1}-G_{t_2}=(\id-G_{t_2-t_1})\circ G_{t_1}
\]
and both maps on the right are strictly increasing (Lemma~\ref{lem:contr}). This also covers the endpoints, with $G_0=\id$ and $G_\infty=0$.

This comparison connects bilateral awards to copulas. Fix a reference claim $\bar c_0>0$. For $x,y\in[0,\bar c_0]$, let $F(x,y)=T_x(y)$, where $T_x$ is the unique gain map with $T_x(\bar c_0)=x$ (Lemma~\ref{lem:common}). A larger $x$ corresponds to a smaller sacrifice. Thus the preceding comparison gives the rectangle inequality
\[
 F(x_2,y_2)-F(x_1,y_2)-F(x_2,y_1)+F(x_1,y_1)\ge0
 \qquad (x_1\le x_2,\ y_1\le y_2).
\]
Since $F(0,y)=F(x,0)=0$, $F(\bar c_0,y)=y$ and $F(x,\bar c_0)=x$, $(u,v)\mapsto F(\bar c_0 u,\bar c_0 v)/\bar c_0$ is a copula on the unit square.

For positive $x,y$, Proposition~\ref{prop:flow} gives
$U(F(x,y))=U(x)+U(y)-U(\bar c_0)$. Thus $F$ is symmetric and associative, with Archimedean generator $U(\bar c_0)-U$. Associativity expresses closure of the gain maps under CD (Lemma~\ref{lem:criterion}(a)). To see how CU gives radial symmetry, index the loss maps by their loss at $\bar c_0$. This gives $D(x,y)=y-F(\bar c_0-x,y)$. By Corollary~\ref{cor:dualflow}, $V(D(x,y))=V(x)+V(y)-V(\bar c_0)$ for positive $x,y$, so $D$ is symmetric as well. Equating $D(\bar c_0-x,y)$ and $D(y,\bar c_0-x)$ and using symmetry of $F$ gives radial symmetry,
\[
 F(x,y)=x+y-\bar c_0+F(\bar c_0-x,\bar c_0-y).
\]
Thus the strictly interior rules of Corollary~\ref{cor:symmetric} yield radially symmetric Archimedean copulas. Among strict Archimedean copulas, radial symmetry selects the Frank family, including the product copula as its limit at parameter zero \citep[Theorem~4.1]{Frank1979}. In the copula formulation the rectangle inequality implies convexity of the additive generator \citep{Nelsen2006}. Here concavity of $U$ follows from the contraction inequality (Lemma~\ref{lem:concave}).
\end{remark}

\subsection{The asymmetric strictly interior case}\label{sec:asym-interior}

Let $r$ satisfy CONS, CD and CU, and restrict it to a priority class $B$ of at least three claimants on which it is strictly interior.

\begin{proposition}[Strictly interior rules without ETE]\label{prop:strict}
On a priority class of at least three claimants, a strictly interior rule satisfying CONS, CD and CU is a log-exponential rule.
\end{proposition}
\begin{proof}
Fix a claimant $j\in B$ and use Lemma~\ref{lem:common} with $\bar c_0=1$. Fix $i$ in the class $B$. On an auxiliary set of three claimants, give each claimant with claim $c_0$ the award $T_{i,a}(c_0)$, choosing $a$ to balance the budget on each nonempty subset. The resulting rule is strictly interior and satisfies ETE and CONS. Since $\T_i$ and $\Lf_i$ are closed under composition, Lemma~\ref{lem:criterion} gives CD and CU. Proposition~\ref{prop:interior} identifies the auxiliary rule as $r^{\theta_i}$ for a finite $\theta_i$. Comparing auxiliary claims $1,c_0$ at the same $a$ gives
\[
 \varphi_{\theta_i}\big(T_{i,a}(c_0)\big)=q_i(a)\varphi_{\theta_i}(c_0),\qquad q_i(a):=\frac{\varphi_{\theta_i}(T_{i,a}(1))}{\varphi_{\theta_i}(1)}.
\]
Each $q_i$ is a continuous strictly increasing bijection of $[0,1]$.

For any $i\in B$, write $q=q_i\circ q_j^{-1}$, so that when the gain map for $j$ multiplies transformed claims by $u$, the gain map for $i$ multiplies them by $q(u)$. Take two indices with multipliers $u,v$ for $j$. Composing the gain maps gives $uv$ for $j$ and $q(u)q(v)$ for $i$. By Lemma~\ref{lem:dualident}, composing the loss maps and taking their complements gives $1-(1-u)(1-v)$ for $j$ and $1-(1-q(u))(1-q(v))$ for $i$. In each case the resulting maps satisfy \eqref{eq:realisability} by Lemma~\ref{lem:criterion}, so the multipliers are related by $q$. Hence
\[
 q(uv)=q(u)q(v),\qquad q(u+v-uv)=q(u)+q(v)-q(u)q(v).
\]
The first identity gives $q(t)=t^\alpha$ for some $\alpha>0$, since $u\mapsto-\log q(e^{-u})$ is continuous and additive on $\R_+$. Letting $u=v=t>0$ in the second identity and dividing by $t^\alpha$ gives $(2-t)^\alpha=2-t^\alpha$. As $t\downarrow0$ this implies $\alpha=1$. Thus $q_i=q_j$ for every $i\in B$. For any problem on $B$, the sum $\sum_{i\in N}T_{i,a}(c_i)$ varies continuously from zero to $\norm{c}$ by Lemma~\ref{lem:common}. Choose $a$ so that the sum equals $E$. Equation~\eqref{eq:realisability} then gives \eqref{eq:rtheta} with $\lambda=q_j(a)$.
\end{proof}

\subsection{The boundary case}\label{sec:mainproof}\label{sec:boundary}

We first show that if a bi-compositional rule has a gain map whose components are proper cap maps, then all of its gain maps are cap maps and the rule necessarily is weighted CEA with the caps as weights. We then show that on a priority class of at least three agents where strict interiority fails, either the rule or its dual has such a gain map. The rule is therefore weighted CEA or, by duality, weighted CEL on that class.

\begin{proposition}[A positive cap determines weighted CEA]\label{prop:cap}
If a rule satisfying CD and CU has a gain map whose components are proper cap maps, then it is weighted CEA with the caps as weights.
\end{proposition}
\begin{proof}
Let $M_s=(M_{i,s_i})_{i\in N^\ast}$ be the gain map formed by the agents' proper cap maps. At claims $c$ the initial awards are $\min\{c_i,s_i\}$. By CU, applying the same caps to the remaining claims in each of $n-1$ further rounds gives total awards $\min\{c_i,ns_i\}$, so $M_{ns}$ is a gain map.

Let $c$ be a nonzero claims vector on $N$ with $c_i\in\{0,ns_i\}$ for each $i\in N$, and let $x=r(c,E)$. The cap $M_{ns}$ selects $c$ from claims $c+x$, since $x_i=0$ wherever $c_i=0$. By CD the award paths for $c+x$ and $c$ agree up to $\norm{c}$. After paying $x$ from $c+x$, the remaining claims are $c$, so CU gives
\[
 r(c,E+F)=r(c,E)+r(c,F),\qquad E,F\ge0,\quad E+F\le\norm{c}.
\]
The award path is continuous by \eqref{eq:endowment-lipschitz} and ends at $c$, so $r(c,E)=Ec/\norm{c}$.

Now let $c\ne0$ be arbitrary and let $h=\min_{i:c_i>0}c_i/s_i$. Choose an integer $n\ge h$ and increase each positive claim $c_i$ by $(n-h)s_i$. The cap $M_{ns}$ then pays $ns_i$ on each positive claim and zero elsewhere. By CD and the preceding paragraph, the award path up to these awards is proportional. Once each agent with a positive claim has received $(n-h)s_i$, the remaining claims are $c$. By CU, the awards for $c$ therefore increase in proportion to the caps until the first positive claim is met, at multiplier $h$.

By CU the same argument applies to the unpaid claims. Thus the awards increase in proportion to the caps among agents whose claims are still unpaid, which is weighted CEA with weights $s_i$.
\end{proof}

Next we show that on a priority class of at least three agents, a rule $r$ satisfying CONS, CU, CD is strictly interior or has a gain map whose components are all proper cap maps or all proper shift maps. Proposition~\ref{prop:cap} gives weighted CEA in the first case and, by duality, weighted CEL in the second.

\begin{proposition}[Boundary allocations]\label{prop:trichotomy}
On each priority class of at least three claimants, a rule satisfying CONS, CD and CU is weighted CEA, weighted CEL, or strictly interior.
\end{proposition}
\begin{proof}
First suppose $B$ is finite. Let $\bar c$ be positive claims on $B$.
By Lemma~\ref{lem:realisability}, there is an endowment at which every
claim is partially satisfied. Let $T$ be the gain map through these awards.
Since $T_i$ is $1$-Lipschitz, $T_i\le\id$ and $T_i(\bar c_i)<\bar c_i$,
its fixed points form an interval $[0,s_i]$ with $s_i<\bar c_i$.
For $c_0\le s_i$, we have $T_i^{\circ n}(c_0)=c_0$ for every $n$.
For $c_0>s_i$, the iterates decrease to a fixed point by continuity,
and monotonicity gives $T_i^{\circ n}(c_0)\ge s_i$. Thus
$T^{\circ n}$ converges pointwise to the cap map $M_s$.
By Lemma~\ref{lem:criterion}, each $T^{\circ n}$ is a gain map.
For any claims $c$, the sum of the awards $T^{\circ n}(c)$ converges to
$\norm{M_s(c)}$. By continuity of the award path, $r$ gives the awards
$M_s(c)$ at that endowment, so $M_s$ is a gain map for $r$. If $s_i>0=s_j$, the cap pays
a positive claim of $i$ in full and gives $j$ nothing, implying absolute
priority by Lemma~\ref{lem:geometry}. Thus the cap map is either proper for all agents or the zero map. In the first case, Proposition~\ref{prop:cap} gives weighted CEA.

Applying the same argument to $\id-T$ gives weighted CEL for $r$ if the cap map for the dual rule is proper for all agents. Otherwise $0<T_i(c_0)<c_0$ for every agent $i$ and every positive claim $c_0$, since full payment by $T_i$ would give $i$ a positive cap for $r$, and a zero award would give $i$ a positive cap for its dual. For any claims $c$,
the awards $T^{\circ n}(c)$ are positive on every positive claim and their
total tends to zero. Endowment monotonicity therefore gives each positive
claim a positive award at every positive endowment. Applying this conclusion
to the dual rule gives positive losses at every positive deficit. Thus $r$
is strictly interior.

Finally in the case that $B$ is infinite, every finite population $N\subset B$ with at least three agents falls under the finite result. Any two such populations have the same classification because both are contained in their finite union, and weighted CEA, weighted CEL and strict interiority are preserved under restriction and are distinct on $N$. Moreover, by CONS, the weights agree (after normalization) across populations, so that there is a global weight profile $(w_i)_{i\in B}$ for the weighted CEA resp. weighted CEL case.
\end{proof}

\subsection{Proof of the characterisation}\label{sec:proof}

\begin{proof}[Proof of Theorem~\ref{thm:main}]
\emph{Sufficiency.} Each log-exponential or weighted rule is consistent
because the same parameter gives every pair its awards at their allotted
total. For the log-exponential rule, \eqref{eq:finite-identities} gives
CD by composition of gain maps and CU by composition of loss maps (Lemma~\ref{lem:criterion}). For
weighted CEA, caps at $w_i s,w_i t$ compose to a cap at $w_i\min\{s,t\}$
and shifts compose to a shift at $w_i(s+t)$. These give CD and CU,
respectively, by Lemma~\ref{lem:criterion}. Weighted CEL follows by duality. Bilateral components
satisfy both composition axioms by assumption, and singletons by
feasibility.

In the priority composition, a pair in one class follows that class's
rule and a pair in distinct classes follows absolute priority. This gives
CONS and identifies the stated classes as the priority classes. At
endowment $E$, the totals $E_B$ pay an initial set of classes in full,
at most one in part, and the rest zero. Dividing $E'\leq E$ among these
totals in the same order gives each class its total under direct division
of $E'$. CD within each class gives the same awards. The dual pays the
classes in reverse order using their dual rules, so its CD gives CU for
the original rule.

\emph{Necessity.} Proposition~\ref{prop:priority} gives the ordered priority
classes. Singletons are determined by feasibility. Two-claimant classes
have precisely the bi-compositional rules of \citet{Chambers2006} for which
the claimants are tied in priority. On a larger class, Proposition~\ref{prop:trichotomy} gives weighted CEA, weighted CEL, or strict interiority. In the last case, Proposition~\ref{prop:strict} gives a log-exponential rule. Proposition~\ref{prop:auto-cont} gives CONT.
\end{proof}

\subsection{Independence of the axioms}\label{sec:independence}

The three axioms of Theorem~\ref{thm:main} are independent: Let the rule use CEA on pairs and proportional division on larger populations. Both composition axioms hold because they compare problems on the same population. At claims $(1,2,3)$ and endowment $3$, the proportional awards are $(1/2,1,3/2)$. Under CEA, the first two claimants would divide their total $3/2$ into $(3/4,3/4)$, so CONS fails. The equal-sacrifice rule for $u(x)=-1/x$ satisfies CONS and CD by \citet{Young1988} and violates CU, as shown in Example~3 of \citet{Moulin2000}. The dual of the preceding rule equalises sacrifices in losses for the same utility. By Lemma~\ref{lem:duality}, it satisfies CONS and CU and violates CD.

\section{Compatibility with other properties}\label{sec:compat}

Since $1+\theta\varphi_\theta(x)=e^{\theta x}$,
\begin{equation}\label{eq:phimult}
\varphi_\theta(a+b)=\varphi_\theta(a)+\varphi_\theta(b)+\theta\,\varphi_\theta(a)\varphi_\theta(b),
\end{equation}
so the transformed claim $\varphi_\theta(c_i)$ is subadditive in the claim for $\theta<0$ and superadditive for $\theta>0$. The gain maps inherit this property for each claimant's parameter.

\begin{proposition}[Subadditivity and superadditivity of the gain maps]\label{prop:subadd}
In a log-exponential rule $r^\theta$, claimant $i$'s gain maps are subadditive if $\theta_i\le0$ and superadditive if $\theta_i\ge0$. For $\theta_i<0$ and $\theta_i>0$, respectively, the inequalities are strict for proper gain maps and positive arguments. The gain maps of weighted CEA are subadditive and those of weighted CEL are superadditive.
\end{proposition}

\begin{proof}
Let $\theta$ be the parameter of a claimant, let $\lambda\in(0,1)$ and $a,b>0$, and write $T_\lambda=\varphi_\theta^{-1}(\lambda\varphi_\theta(\cdot))$. Applying \eqref{eq:phimult} to $T_\lambda(a)+T_\lambda(b)$ and to $a+b$, and using $\varphi_\theta(T_\lambda(c_0))=\lambda\varphi_\theta(c_0)$, we obtain
\[
\varphi_\theta\big(T_\lambda(a)+T_\lambda(b)\big)-\varphi_\theta\big(T_\lambda(a+b)\big)
=\theta\lambda(\lambda-1)\,\varphi_\theta(a)\varphi_\theta(b),
\]
which is positive for $\theta<0$ and negative for $\theta>0$, since $\lambda(\lambda-1)<0$. As $\varphi_\theta$ is strictly increasing, the claim follows for the proper gain maps. The maps $T_0=0$ and $T_1=\id$ are additive. For weighted CEA and weighted CEL, the inequalities $\min(a+b,s)\le\min(a,s)+\min(b,s)$ and $(a+b-s)_+\ge(a-s)_++(b-s)_+$ apply with the cap and shift $w_i s$, respectively.
\end{proof}

A rule is {\bf merging-proof} if a group can never gain by consolidating its claims. Whenever a problem is obtained from $(N,c,E)$ by replacing a group $A\subseteq N$ with a single claimant of claim $\sum_{i\in A}c_i$, that claimant's award is at most $\sum_{i\in A}r_i(c,E)$. A rule is {\bf splitting-proof} if that award is always at least $\sum_{i\in A}r_i(c,E)$.\footnote{In the claims literature these properties are usually called \emph{no advantageous merging} and \emph{no advantageous splitting}, or non-manipulability via merging and via splitting; see \citet{ChambersThomson2002}.} The claimant with the combined claim may be any member of $A$. \citet{deFrutos1999} characterises the parametric rules satisfying each property, and \citet{Ju2003} gives characterisations for pairwise merging and splitting under ETE, consistency and continuity. In our symmetric family, subadditivity of the gain maps gives merging-proofness and superadditivity gives splitting-proofness.

\begin{corollary}[Merging and splitting]\label{cor:compat}
Among the rules satisfying CONS, CD and CU, the merging-proof rules are exactly the symmetric $r^\theta$ with $\theta\in[-\infty,0]$, and the splitting-proof rules exactly the symmetric $r^\theta$ with $\theta\in[0,+\infty]$. $\PROP$ is the only rule satisfying both.
\end{corollary}

\begin{proof}
First, either property implies ETE. Let $i,j,k$ be distinct with claims $(0,a,b)$ and fix an endowment $E$. Feasibility gives $i$ zero, so $j$ and $k$ divide all of $E$. By CONS the award of $j$ is $r_j^{jk}((a,b),E)$. Merging $i$ and $j$ into $i$ gives $r_i^{ik}((a,b),E)$. Under merging-proofness the latter is at most the former, and under splitting-proofness it is at least the former. Starting instead with claims $(a,0,b)$ and merging into $j$ gives the opposite inequality under each property. Hence either property allows us to replace one claimant in a pair by a third without changing the awards in the corresponding coordinates. At claims $(a,a)$, replace $(i,j)$ successively by $(k,j)$, $(k,i)$ and $(j,i)$. The first award is unchanged at each step. It starts as $i$'s award and ends as $j$'s award in the same bilateral problem, so the two are equal. By CONS this gives ETE on every population, and Corollary~\ref{cor:symmetric} gives a symmetric $r^\theta$.

Let $\theta\in[-\infty,0]$ and let $T$ be a gain map that solves $(N,c,E)$. If the claimants in $A$ merge, subadditivity gives $T(\sum_{i\in A}c_i)\le\sum_{i\in A}T(c_i)$, so the sum of the new awards at $T$ is at most $E$. The gain maps are ordered pointwise. If the sum is less than $E$, the gain map that solves the merged problem must be larger than $T$. If it equals $E$, we can still use $T$. Every claimant outside $A$ therefore receives weakly more, and budget balance leaves the merged claimant with at most the original total award to $A$. For $\theta\in[0,+\infty]$ the inequalities reverse, which is splitting-proofness.

For the converse, take claims $(1,1,2)$ and $E=2$, and let $T_\lambda$ be the proper gain map solving this problem. Merging the first two claimants gives claims $(2,2)$ and awards $(1,1)$ by ETE. If $\theta>0$ is finite, strict superadditivity gives $T_\lambda(2)>2T_\lambda(1)$. Since the original budget equation is $2T_\lambda(1)+T_\lambda(2)=2$, it follows that $2T_\lambda(1)<1$, so merging the two unit claims is profitable. If $\theta<0$ is finite, strict subadditivity gives $T_\lambda(2)<2T_\lambda(1)$, and the same budget equation yields $2T_\lambda(1)>1$, so splitting a claim of size $2$ into two unit claims is profitable.

At the endpoints,
\[
\CEL((1,1,2),2)=\left(\tfrac13,\tfrac13,\tfrac43\right),
\qquad
\CEL((2,2),2)=(1,1),
\]
so merging the two unit claims is profitable under $\CEL$. Likewise,
\[
\CEA((1,1,2),2)=\left(\tfrac23,\tfrac23,\tfrac23\right),
\qquad
\CEA((2,2),2)=(1,1),
\]
so splitting a claim of size $2$ into two unit claims is profitable under $\CEA$. Thus no additional members are merging-proof or splitting-proof. Finally $[-\infty,0]\cap[0,+\infty]=\{0\}$, and $\PROP$ satisfies both properties.
\end{proof}

Two further properties select the weighted versions of the endpoints within each priority class. {\bf Invariance under claims truncation} requires $r(c,E)=r\big(t(c,E),E\big)$ with $t_i(c,E):=\min\{c_i,E\}$. {\bf Minimal rights first} requires $r(c,E)=m(c,E)+r\big(c-m(c,E),E-\sum_i m_i(c,E)\big)$ with $m_i(c,E):=\big(E-\sum_{j\ne i}c_j\big)_+$.

\begin{proposition}[Claims truncation and minimal rights]\label{prop:truncation}
Among the rules satisfying CONS, CD and CU, invariance under claims truncation holds exactly for priority compositions of weighted CEA rules, and minimal rights first holds exactly for priority compositions of weighted CEL rules.
\end{proposition}
\begin{proof}
Fix a pair and suppose the rule is invariant under claims truncation. Let $g(E)=r((E,E),E)$, so $g(E)$ gives the awards whenever both claims are at least $E$. At claims $(E+F,E+F)$, the awards at endowment $E$ are $g(E)$ and both remaining claims are at least $F$. By CU and truncation, $g(E+F)=g(E)+g(F)$. Continuity gives $g(E)=Ew$ for a vector $w\ge0$ with $\sum_i w_i=1$.

For arbitrary positive bilateral claims, whenever both claims are unmet, CU gives an increment $Fw$ for any increase $F$ no larger than the smaller remaining claim. Thus the awards grow in proportion to $w$ until a claim is met. Any further endowment goes to the other claimant. A zero weight gives the other claimant absolute priority, so the two belong to different priority classes. Within any one class both weights are therefore positive, giving weighted CEA on each pair. On a class of at least three claimants, Theorem~\ref{thm:main} then gives weighted CEA: a weighted CEA pair can pay one claim in full at an interior endowment, ruling out a log-exponential rule, and gives both claimants positive awards at every positive endowment for positive claims, ruling out weighted CEL. Feasibility covers singletons.

Conversely, truncation leaves unchanged the claims in any class that is paid in full, since each such claim is at most $E$. If a class receives a positive endowment below its total claims, the truncated claims in that class still sum to at least its allotted endowment. This class therefore keeps its endowment, and every lower class still receives zero. For weighted CEA, each award $x_i=\min\{c_i,w_i s\}$ satisfies $x_i\le E$. Thus replacing $c_i$ by $\min\{c_i,E\}$ leaves the award at the same $s$ unchanged.

Minimal rights first is the dual of invariance under claims truncation, since truncating at the deficit $\norm{c}-E$ removes $m_i(c,E)$ from claim $i$. Duality reverses the order of the classes and interchanges weighted CEA and weighted CEL, giving the second claim.
\end{proof}

\bibliographystyle{plainnat}
\bibliography{references}

\appendix
\section{Endowment monotonicity and consistency}\label{app:foundations}

This appendix records two consequences of the axioms for finite participating populations of any size. Either composition axiom implies endowment monotonicity (Lemma~\ref{lem:resmon}). Under endowment monotonicity, bilateral consistency implies consistency over subgroups (Lemma~\ref{lem:fullcons}). We use endowment monotonicity throughout. Consistency over subgroups enters the continuity proof and determines the division within each priority class.

\begin{lemma}[Endowment monotonicity; Moulin \citeyearpar{Moulin2000}, p.~653]\label{lem:resmon}
Either composition axiom implies endowment monotonicity: for $E' \le E$, $r(c,E') \le r(c,E)$ coordinatewise.\end{lemma}

\begin{proof}
Under CD, $r(c,E')=r(r(c,E),E')\le r(c,E)$ by claims boundedness. Under CU the difference $r(c,E)-r(c,E')$ is an awards vector and is nonnegative.
\end{proof}

\begin{lemma}[Consistency from bilateral consistency; Chun \citeyearpar{Chun1999}]\label{lem:fullcons}
Under endowment monotonicity, the bilateral rules determine an allocation uniquely at each endowment. Consequently \textup{CONS} implies \textup{CONS}$^\ast$.
\end{lemma}

\begin{proof}
If two vectors $x,y$ with the same total differ, there are $i,j$ with $x_i>y_i$ and $x_j<y_j$. Their projections cannot both lie on the same bilateral award path, by endowment monotonicity. Now let $x=r(N,c,E)$ and $\varnothing\ne A\subseteq N$. By CONS, both $x_A$ and $r(A,c_A,\sum_{i\in A}x_i)$ satisfy all the bilateral rules on $A$. They have the same total and therefore coincide. For singletons the conclusion follows from feasibility.
\end{proof}

\section{Lorenz ranking}\label{app:Lorenz}
\begin{proof}[Proof of Proposition~\ref{prop:lorenz}]
\emph{Two claimants, finite parameters.} Let $\theta<\theta'$ be finite. For $u_\theta=\log\varphi_\theta$ we have $u_\theta'=e^{\theta x}/\varphi_\theta(x)$ and
\[
 \frac{u_\theta''(x)}{u_\theta'(x)}=-\frac1{\varphi_\theta(x)}< -\frac1{\varphi_{\theta'}(x)}=\frac{u_{\theta'}''(x)}{u_{\theta'}'(x)},\qquad x>0,
\]
since $\varphi_\theta(x)=\int_0^x e^{\theta t}\,dt$ is strictly increasing in $\theta$. Thus $u_\theta\circ u_{\theta'}^{-1}$ is strictly increasing and strictly concave on the range of $u_{\theta'}$.

Take $c_1>c_2>0$ and $0<E<\norm{c}$. Let $x=r^\theta(c,E)$ and $x'=r^{\theta'}(c,E)$. Equal sacrifice for $u_{\theta'}$ gives
\[
 u_{\theta'}(x_1')-u_{\theta'}(x_2')=u_{\theta'}(c_1)-u_{\theta'}(c_2)>0.
\]
Since $u_{\theta'}$ is increasing, $x'_1>x'_2$. The two intervals in utility units have the same length, and the first lies further to the left since $x_1'<c_1$. Strict concavity of $u_\theta\circ u_{\theta'}^{-1}$ therefore gives
\[
 u_\theta(x_1')-u_\theta(x_2')>u_\theta(c_1)-u_\theta(c_2)=u_\theta(x_1)-u_\theta(x_2).
\]
At a fixed total $E$ this difference increases strictly with the larger award. Hence $x_1'>x_1$, so $x_2>x_2'$. Interchanging the claimants gives the same conclusion for $c_2>c_1$. Equal claims, a zero claim or an extreme endowment give identical awards.

\emph{Two claimants, endpoints.} The Lorenz inequalities pass to coordinatewise limits. By Proposition~\ref{prop:endpoints}, $\CEA$ Lorenz dominates every finite member, and each finite member Lorenz dominates $\CEL$. For unequal positive claims and an interior endowment, strictness when one or both parameters are infinite follows by taking two distinct finite parameters in the interval between them and applying the strict inequality just proved.

\emph{Arbitrarily many claimants.} Each symmetric rule is consistent because the common parameter still balances the budget of every pair at its allotted total. At fixed claims, the awards are weakly increasing in the endowment. At a fixed parameter, the gain maps are weakly increasing in the claim, including the cap and shift maps for $\CEA$ and $\CEL$. Theorem~5 of \citet{Thomson2012lorenz} therefore lifts bilateral Lorenz dominance to every finite population.

\emph{Strictness for arbitrarily many claimants.} Let $\theta<\theta'$ and suppose the claims include two distinct positive claims and $0<E<\norm{c}$. Let $x=r^\theta(c,E)$ and $x'=r^{\theta'}(c,E)$. Both award vectors are ordered by the same claims, so equality of all Lorenz sums implies $x=x'$. If at least one parameter is finite, choose $i,j$ with distinct positive claims. Their claims are both partially satisfied under the finite-parameter rule, so $0<x_i+x_j<c_i+c_j$. By CONS, equality of the two award vectors gives identical bilateral awards at that endowment, contradicting strict bilateral dominance. For the two endpoints, choose a finite parameter between them. Every Lorenz sum for $\CEA$ is at least the corresponding sum for this finite member, with strict inequality for some $k$. Its $k$th sum in turn is at least that for $\CEL$.
\end{proof}

\section{Gain maps when every claim is partially satisfied}\label{app:gain-maps}
\begin{lemma}[Gain maps when every claim is partially satisfied]\label{lem:realisability}
Let $r$ satisfy CONS, CD and CU. Consider a problem $(N,c,E)$ with
$|N|\ge3$, and write $x=r(N,c,E)$. If $0<x_i<c_i$ for every $i\in N$,
there is a unique gain map $T\in\T$ with $T(c)=x$. Each coordinate $T_i$
is nondecreasing and $1$-Lipschitz.

If $c>0$ and the claimants in $N$ are tied in priority, there exists
$E'\in(0,\norm{c})$ such that $0<r_i(N,c,E')<c_i$ for every $i\in N$.
\end{lemma}

\begin{proof}
Suppose $0<x_j<c_j$ for all $j\in N$. For any claimant $i$, choose
$j\in N\setminus\{i\}$. Define $T_i(c_0)$ as $i$'s award on the pair
with claims $(c_0,c_j)$ when $j$ receives $x_j$. This award is unique
by Lemma~\ref{lem:geometry} for $c_0>0$, and is zero for $c_0=0$.
Any gain map through $x$ must give this award by \eqref{eq:realisability}.

To compare two choices $j,k$, give the triple $i,j,k$ claims $(c_0,c_j,c_k)$
and choose an endowment at which $j$ receives $x_j$. By CONS the awards
to $j,k$ lie on their original bilateral path, where $x_j$ determines
$x_k$. Applying CONS to $i,j$ and $i,k$ gives the same award to $i$.
In particular, $T(c)=x$ by CONS in the original problem.

Now give any pair $i,j$ arbitrary claims $d_i,d_j$. Since $|N|\ge3$,
we can add $k\in N\setminus\{i,j\}$ with claim $c_k$. When $k$ receives
$x_k$, the other two receive $T_i(d_i),T_j(d_j)$. By CONS their bilateral
rule selects these awards at their sum. Bilateral determination
(Lemma~\ref{lem:fullcons}) gives \eqref{eq:realisability} on every population.

Increasing $c_0$ weakly lowers $j$'s award at a fixed endowment
(Proposition~\ref{prop:auto-cont}). Hence the endowment $x_j+T_i(c_0)$
at which $j$ receives $x_j$ weakly increases. Applying the same argument
to the dual rule with $j$'s loss fixed at $c_j-x_j$ shows that
$c_0-T_i(c_0)$ is also nondecreasing. Thus $T_i$ is nondecreasing and
$1$-Lipschitz.

Finally, suppose $c>0$ and the claimants in $N$ are tied in priority.
Let $I_i$ be the nonempty open interval of endowments at which $i$'s
claim is partially satisfied. If $\sup I_i\le\inf I_j$, then at $\sup I_i$ claimant
$i$ is fully paid and $j$ receives zero. By CONS and Lemma~\ref{lem:geometry}
this implies $i\succ_r j$. Thus the intervals intersect pairwise and,
since $N$ is finite, have a common endowment $E'$.

\end{proof}

\end{document}